\documentclass[12pt]{iopart}
\usepackage{amssymb,iopams,bm,verbatim,enumerate,perpage}
\usepackage[T1]{fontenc}
\usepackage{amsthm}
\usepackage{color}
\usepackage[hidelinks]{hyperref}
% ...
\MakePerPage[2]{footnote}
\def\Journal#1#2#3#4{{#4} {\it #1} {\bf #2}, #3 }

\def\ud{\textrm{d}}

\newcommand{\smfrac}[2]{{\textstyle{#1\over#2}}}
\def\half{\smfrac{1}{2}}
\newcommand{\w}[1]{\bm{#1}} 
\def\h{h}
\def\ud{\textrm{d}}
\def\u{\dot{u}}
\def\U{\dot{U}}
\def\modU{\textrm{U}}
\def\B{\mathfrak{X}} 
\def\E{\mathfrak{E}} 
\def\Q{\mathfrak{Q}} 
\def\R{\mathfrak{R}} 
\def\b{\mathfrak{b}} 
\def\O{\mathfrak{o}} 
\def\N{\mathfrak{N}} 
\def\J{\mathfrak{J}} 
\def\X{\w{X}} 
\def\Y{\w{Y}} 
\def\Z{\w{Z}} 
\def\tfrac{\smfrac}
\newcommand{\HH}{\mathcal{H}}

\newtheorem*{de}{Definition}
\newtheorem{re}{Remark}
\newtheorem{lm}{Lemma}
\newtheorem{te}{Theorem}

\newtheorem{co}{Corollary}
\begin{document}

\title[Shear-free perfect fluids]{Shear-free perfect fluids with a barotropic equation of state in general relativity: the present status
}
\author{Norbert Van den Bergh$^1$ and Radu Slobodeanu$^{2,3}$}
\address{$^1$\ Faculty of Engineering and Architecture FEA16, Gent University, Galglaan 2, 9000 Gent, Belgium}
\address{$^2$ \ Faculty of Physics, University of Bucharest, P.O. Box Mg-11, RO-077125 Bucharest-Magurele, Romania}
\address{$^3$ \ Institute of Mathematics, University of Neuch\^atel, 11 rue Emile Argand, 2000 Neuch\^atel, Switzerland}
\eads{\mailto{norbert.vandenbergh@gmail.com}, \mailto{radualexandru.slobodeanu@g.unibuc.ro}}

\begin{abstract}
The present status of the shear-free fluid conjecture in general relativity is discussed: a review is given of recent partial proofs, a new and complete
proof is given for the case of a linear equation of state, \emph{including a non-zero cosmological constant}, and a number of useful results are presented
which might help in proving the conjecture for a general equation of state.
\end{abstract}


\section{Introduction}
We consider perfect fluid solutions of the Einstein field equations,
\begin{equation}
R_{ab} - \tfrac{1}{2} R \, g_{ab} = T_{ab}
, \label{EFE}
\end{equation}
on a 4-dimensional spacetime $(M, g)$, with energy-momentum tensor given by
\begin{equation}\label{intro_Tab}
T_{ab}= (\mu+p) u_{a}u_{b}+p g_{ab},
\end{equation}
$\mu$ and $p$ being respectively the energy density and pressure of the fluid and the unit time-like vector field $u_a$ being the fluid's (covariant) 4-velocity. As is well known, 
the covariant derivative of $u_a$ can be decomposed as
\begin{equation}\label{uadecomp}
u_{a;b}=\tfrac{1}{3}\theta (g_{ab}+u_a u_b)+\sigma_{ab}+\omega_{ab}-\dot{u}_{a}u_{b},
\label{eq2}
\end{equation}
where $\theta$ is the fluid's (rate of volume) expansion, $\dot{u}_{a}$ is the acceleration and $\sigma_{ab}$, $\omega_{ab}$ are respectively the
shear and vorticity tensors, which are uniquely defined by (\ref{eq2}) and the properties
\begin{equation}
 u^a \dot{u}_{a} = 0,\ u^a\omega_{ab} = \ u^a\sigma_{ab} =0, \ \sigma_{[ab]}=\omega_{(ab)}=0, \ {\sigma^a}_a=0.
\end{equation}
The physical significance of these so called \emph{kinematical quantities} has been discussed by many authors, see for example \cite{Ellis1971}. Among
well known explicit solutions of the Einstein field equations \cite{Kramer}, in which some of these quantities vanish, we note the following \textit{shear-free} ($\sigma_{ab}=0$) solutions: the Einstein static 
universe ($\theta=\dot u_a=\sigma_{ab}=\omega_{ab}=0$), FLRW universes ($\theta \neq 0$, $\dot u_a=\sigma_{ab}=\omega_{ab}=0$) and the G\"odel universe ($\theta = \dot u_a = \sigma_{ab}=0$, $\omega_{ab} \neq 0$). 
Imposing a \textit{barotropic equation of state} $p = p(\mu)$, a common feature of the above examples, leads to extra restrictions on the solution space: for example all barotropic and shear-free ($\sigma_{ab}=0$) perfect fluids with non-vanishing expansion and vanishing vorticity are known explicitly \cite{CollinsKV}. This is not the case for barotropic and shear-free perfect fluids with vanishing expansion and non-vanishing vorticity \cite{Karimian2012}, although also here large classes of solutions 
exist (for example all rigidly rotating axisymmetric and stationary perfect fluids belong to this family). Remarkably, barotropic and shear-free perfect fluids in which both expansion and vorticity are non-zero seem to be confined to the limiting situation of `$\Lambda$-type' models (meaning $p=-\mu=constant$), the only example known to us being a Bianchi IX model found by Obukhov et al.~\cite{Obukhov}. This  brings us to the subject of the present paper, the so-called \emph{shear-free fluid conjecture} which claims that
\begin{quote}
general relativistic, shear-free perfect fluids obeying a barotropic equation of state such that $p + \mu \neq  0$, are either non-expanding or non-rotating.
\end{quote}
During the last few years there has been a renewed interest in this conjecture, which, if true, would be a remarkable consequence of the full Einstein
field equations: on the one hand Newtonian perfect fluids with a barotropic equation of state, which are rotating and expanding but non-shearing, are known
to exist \cite{Ellis2011,HeckmannSchucking, Narlikar,SenSopSze}, while on the other hand
in, for example, $f(R)$ gravity there is no counterpart of the conjecture either \cite{Sofuoglu}.

The first suggestion that the vanishing of shear
could play a decisively restricting role in the construction of expanding and rotating perfect fluids appeared in 1950, without proof, in a somewhat obscure contribution by
G\"odel \cite{Godel2} on homogeneous rotating cosmological models. A precise formulation of G\"odel's claim was given in 1957 by Sch\"ucking \cite{Schucking}, who gave a short coordinate based
proof that spatially homogeneous dust models ($p=0$) could be either rotating or expanding, but not both.  The condition of vanishing pressure was dropped by
Banerji \cite{Banerji}, who gave a similar coordinate based proof for (tilted) spatially homogeneous perfect fluids obeying a $\gamma$-\textit{law equation of state},
$p=(\gamma-1)\mu$, with $\gamma - 1 \neq \frac{1}{9}$ \footnote{the reason why $\gamma - 1 = \frac{1}{9}$ is special was clarified in \cite{NorbertCQG1999}, where also a proof was given for non-spatially homogeneous spacetimes}. Sch\"ucking's result was generalized in 1967 by Ellis \cite{Ellis1}, who used
the orthonormal tetrad formalism to show that the restriction of spatial homogeneity was redundant for dust spacetimes (in \cite{WhiteCollins} it
was observed that Ellis'
result remained valid in the presence of a cosmological constant).
In \cite{TreciokasEllis} Treciokas and Ellis proved, again using a combination of an orthonormal tetrad formalism and an adapted choice of coordinates,
that the conjecture held true also for the equation of state $p=\frac{1}{3} \mu$, a result which was generalised by Coley \cite{Coley} to allow for a possible non-zero
cosmological constant. In \cite{TreciokasEllis} an outline of an argument was presented, indicating the validity of the conjecture for perfect fluids
in which the \textit{acceleration potential} $r = \exp \int_{p_0}^p \frac{1}{\mu+p} \ud p$ satisfies an equation of the form $\dot{r}= \beta(r)$, where the \textit{dot-operator} is the derivative along the fluid 4-velocity. This result
(which implies the validity of the
conjecture for a general equation of state, once one additionally assumes spatial homogeneity, as was the case in \cite{Banerji, KingEllis, Whiteth}) will play a key role in the sequel.
However the details of the underlying proof remained veiled until 1988, when Lang and Collins \cite{Lang,Langth} explicitly showed that $\omega \theta=0$ indeed follows,
provided there exists a functional relation of the form $\theta=\theta(\mu)$ (which, by the conservation law $\dot{\mu} + (\mu+p) \theta=0$, is equivalent with $\dot{r}=\beta (r)$). 
A `covariant' proof
of this same result was given by Sopuerta in \cite{Sopuerta1998}.
While Treciokas and Ellis already questioned the possible existence of rotating and expanding perfect fluids with
$p=p(\mu)$, their non-existence was explicitly conjectured by Collins \cite{CollinsKV}, following a series of papers in which the conjecture was proved successively for the cases where
the vorticity vector is parallel to the acceleration (see \cite{WhiteCollins}, or \cite{SenSopSze} for a fully covariant proof), or in which the Weyl tensor is purely magnetic \cite{Collins1984} or purely electric \cite{Langth, CyganowskiCarminati}.

Since then the conjecture has been proved also in a large number of special cases, such as
$\ud p/\ud \mu = -\frac{1}{3}$ \cite{CyganowskiCarminati, Langth, slob, NorbertCQG1999}; $\theta=\theta(\omega)$ \cite{Sopuerta1998};
Petrov types N \cite{Carminati1990} and III \cite{CarminatiCyganowski1996,CarminatiCyganowski1997}; the existence of a conformal Killing vector parallel to the fluid flow \cite{Coley}; the Weyl tensor having either a divergence-free electric part \cite{NorbertKarimianCarminatiHuf2012}, or a divergence-free magnetic
part, in combination with an equation of state which is of the $\gamma$-law type \cite{NorbertCarminatiKarimian2007} or which is sufficiently
generic \cite{CarminatiKarimianNorbertVu2009}, and in the case where the Einstein field equations are linearised
about a FLRW background \cite{Nzioki} . 
A major step has been achieved recently by the second author \cite{Slobodeanu2014} proving the conjecture for an arbitrary $\gamma$-law equation of state
(except for the cases $\gamma-1 = -\smfrac{1}{5}, -\smfrac{1}{6},-\smfrac{1}{11},-\smfrac{1}{21},\smfrac{1}{15}, \smfrac{1}{4}$) and a vanishing cosmological constant. In this approach, reminiscent of Pantilie's classification result on
Einstein manifolds \cite{Pantilie2002}, the Einstein field equations were seen as a second order differential system in the length scale function,
with the integrability conditions for this system allowing one
to prove the conjecture \emph{via} some sufficient conditions in terms of \emph{basic functions}, i.e.~functions
that are constant along the fluid flow.\\
Finally, in a recent paper by Carminati~\cite{Carminati2015}, an attempt of a proof was given for a linear equation of state and vanishing
cosmological constant. However this proof is invalid, as inappropriate use was made of Maple's \texttt{solve} command, which for parametric polynomial systems only returns generic
solutions\footnote{namely solutions valid in an open set of the parameter space, i.e.~\texttt{solve(a*x = 0,x)} will only return
\texttt{x = 0}; a bug in Maple's  \texttt{solve} code (Maple support, private communication) also prevents the issuing
of a warning message that solutions might have been lost, even if the \texttt{parametric = full} option is used.}.
Furthermore, the set of equations used in \cite{Carminati2015}
was under-determined, a fact made obvious by inspection of the special case in which there is a Killing vector along the
vorticity, leading to the simplifications $\textrm{u3 = f3 = T13 = T23 = g3 = m3 = n = 0}$.\\

In the present paper we will complete the proof of \cite{Slobodeanu2014}, covering the exceptional values of $\gamma$ and allowing also for a non-zero cosmological
constant, or, equivalently, generalising the equation of state to the form $ p=(\gamma -1) \mu + p_0$.
Inclusion of the constant term is important, first of all as the analysis of the conjecture for a general
equation of state in a natural way is split into two branches, either $p''=\ud^2 p / \ud \mu^2 =0$, leading to $ p=(\gamma -1) \mu + p_0$, or to $p'' \neq 0$, with the analysis
of the second case heavily leaning on the former. A second justification for including the $p_0$ term is that the dimension of a solution space of a set of exact solutions of the Einstein field equations, obtained by imposing kinematic restrictions, may change drastically by the inclusion of
a non-zero cosmological constant. A typical example is provided by the Petrov type I \emph{silent universes}, for which the orthogonal spatially homogeneous Bianchi type I metrics most likely \cite{Sopuerta1997} are the only admissible metrics when $\Lambda=0$, but which for $\Lambda >0$ have been shown \cite{NorbertWylleman2004} to contain a peculiar set of non-OSH models.\\
In addition we generalize the formalism of \cite{Slobodeanu2014} to a general equation of state and we present some theorems, which not only will play a key role in the present proof for a linear equation of state, but which likely will also be useful when tackling
the conjecture in its full generality, when $p=p(\mu)$ is an arbitrary function ($p \neq-\mu$) of the matter density. These
theorems tell us that the conjecture is valid
provided certain algebraic restrictions are obeyed by the kinematical quantities, or that, if the conjecture does \emph{not} hold, there exists a Killing vector along the vorticity. In the latter
case the equations describing the problem simplify dramatically, but the accompanying loss of information turns this sub-case, as remarked already by
Collins \cite{CollinsKV}, into an exceptionally elusive one. The simplest of these
criteria (Corollary \ref{co1}) says that, for an expanding and rotating perfect fluid obeying a barotropic equation of state, the existence of a Killing vector along the vorticity is equivalent with
the acceleration being orthogonal with the vorticity.\\ 
We begin with introducing in section 2 the necessary notations and conventions, while in section 3 we
make the link with the formalism used in \cite{Slobodeanu2014} and present the governing equations for the case of an arbitrary equation of state. In section 4 we prove the general theorems mentioned above. In section 5 we prove the conjecture for the case of a linear equation of state, by splitting the argument according to whether the acceleration is orthogonal to vorticity or not, in 
Theorems \ref{Theorem4} and \ref{Theorem5}. The last sections are dedicated to conclusions and technical Appendices.

\section{Notations and fundamental equations}

We introduce at each point of spacetime an orthonormal tetrad $(\w{e}_a)=(\w{e}_0, \w{e}_\alpha)$ with 
the time-like unit vector $\w{e}_0$ 
coinciding with the fluid 4-velocity $\w{u}$ (henceforth 
Latin indices are tetrad indices taking the values 0,1,2,3, while Greek indices are spatial triad indices taking the values 1, 2, 3). 
Boldface symbols always 
will refer to vector (tensor) fields, but for readability (and as is customary in the literature, see e.g.~\cite{EllisMaartensmacCallum2012}) we 
will also write
$\w{e}_a=\partial_a$: for example $\w{u}=\partial_0$, $\dot{\w{u}}=\dot{u}^\alpha \partial_\alpha$, $\dot{\w{u}}^2 = \dot{u}_\alpha \dot{u}^\alpha$ etc. ... \\
The volume 4-form components will be denoted by $\eta_{abcd}$ with the convention $\eta_{0123}=-1$; its restriction to tangent hyperplanes orthogonal 
to $\w{u}$ is $\varepsilon_{\alpha\beta\gamma}$. 
To a space-like 2-form one associates a vector field by Hodge duality, e.g.~the vorticity vector
$\w{\omega}$ has components $\omega_\alpha = \tfrac{1}{2}\varepsilon_{\alpha\beta\gamma}\omega^{\beta\gamma}$.  The notation $\omega$ will stand for the norm of the vorticity vector / 2-form.

To fix the sign conventions let us point out that the metric components are $(g_{ab})=\mathrm{diag}(-1,1,1,1)$ and that the Riemann and Ricci curvature tensors respectively satisfy
\begin{equation}\label{ricciv}
{v^a}_{;d;c}-{v^a}_{;c;d} = {R^a}_{bcd}v^b \ , \quad  R_{ab}={R^m}_{amb},
\end{equation}
while the 'trace-free part' of the curvature, given by the Weyl tensor, is
\begin{equation}\label{Weyldef}
C_{abcd}=R_{abcd} - (g_{a[c}R_{d]b}+g_{b[d}R_{c]a})+\tfrac{1}{3}R \, g_{a[c}g_{d]b}.
\end{equation}

\paragraph{An extended tetrad formalism.} We will use the extended orthonormal tetrad formalism \cite{EllisMaartensmacCallum2012, Norbert2013}, in which the 
\textit{main variables} are
\begin{itemize}
\item the tetrad basis vectors $\partial_a$, 

\item the kinematical quantities $\dot{u}_\alpha$, $\omega_\alpha$, $\theta$, $\sigma_{\alpha \beta}$, 

\item the local angular velocity $\Omega_\alpha$ of the triad $\partial_{\alpha}$ with respect to a set of Fermi-propagated axes and the Kundt-Sch\"ucking-Behr variables \cite{MacCallum1971} $a_{\alpha}$ and $n_{\alpha \beta}=n_{\beta \alpha}$ which parametrize the purely spatial commutation coefficients ${{\gamma}^\alpha}_{\beta\kappa}$. They are defined by the relations (see \cite{vanElstPHD} for more explicit formulae)
\begin{eqnarray}\label{comm_explicit}
{}[\partial_0,\partial_\alpha] &=& \dot{u}_\alpha \partial_0 - \left(\smfrac{1}{3} \theta\delta_\alpha^{\beta}+\sigma_\alpha^{\beta}
+{\varepsilon^\beta}_{\alpha\gamma}(\omega^\gamma+\Omega^\gamma)\right) \partial_\beta~, \\
{}[\partial_\alpha, \partial_\beta ] &=& {{\gamma}^c}_{\alpha\beta}\partial_c\equiv -2\varepsilon_{\alpha\beta\gamma}\omega^\gamma\partial_0 + \left(2a_{[\alpha}\delta^\gamma_{\beta]}+
\varepsilon_{\alpha\beta\delta} n^{\delta\gamma}\right)\partial_\gamma~.\label{comm_explicit_bis}
\end{eqnarray}
Sometimes, it is computationally advantageous to replace $a_{\alpha}$ and $n_{\alpha\beta}$ $(\alpha\neq\beta)$ with new variables $q_{\alpha}$ and $r_{\alpha}$ defined by 
\begin{equation*}
n_{\alpha-1 \,\alpha+1}=(r_{\alpha}+q_{\alpha})/2, \quad a_{\alpha}=(r_{\alpha
}-q_{\alpha})/2.
\end{equation*}

\item the energy density $\mu$ and pressure $p$, 

\item the `electric' and `magnetic' parts $E_{\alpha \beta}$, $H_{\alpha \beta}$ of the Weyl tensor with respect to $\w{u}$:
\begin{equation}
 E_{ab}= C_{acbd}u^c u^d , \quad 
 H_{ab}= \frac{1}{2} \eta_{amcd} {C^{cd}}_{bn}u^m u^n. \label{EHdef}
\end{equation}
They are symmetric trace-free tensors that determine the Weyl curvature. 
\end{itemize}
In addition, we shall use the following \textit{auxiliary variables}: the spatial gradient of the expansion scalar, $z_\alpha=\partial_\alpha 
\theta$, and the (covariant) divergence of the acceleration, $ j\equiv {\dot{u}^a}_{;a} =
\partial_{\alpha}\dot{u}^{\alpha}+\dot{u}^{\alpha}\dot{u}_{\alpha}-2 \dot{u}^{\alpha} a_{\alpha}$.\\

Note that with this choice of variables, once we assume that the Einstein equations (\ref{EFE}) are satisfied, the Riemann tensor is 
actually \emph{defined} in terms of $\w{E},\w{H},p$ and $\mu$ via (\ref{Weyldef}, \ref{EHdef}) \footnote{for example $R_{1212}=\smfrac{1}{3} 
\mu - E_{33}$. See \cite{vanElstPHD} for a full list of such relations.}, with the symmetry 
and trace-free properties of $\w{E}$ and $\w{H}$ guaranteeing the usual symmetry properties of a curvature tensor. 
The usual defining formulae (obtained from the second Cartan structure equations or, equivalently, from (\ref{ricciv})),
\begin{equation}\label{cartan2bis}
{R^a}_{bcd}={\Gamma ^a}_{bd,c}-{\Gamma ^a}_{bc,d}+{\Gamma ^e}_{bd}
{\Gamma ^a}_{ec} - {\Gamma ^e}_{bc} {\Gamma ^a}_{ed} - {{\gamma}^e}_{cd} {\Gamma ^a}_{be}~,
\end{equation}
become then a set of first order partial differential \emph{equations} in the connection coefficients
${\Gamma ^c}_{ab}$, which are related to the main variables of the formalism through the commutation coefficients: 
\begin{equation}\label{commdef2}
\Gamma_{\ ab}^c = \tfrac{1}{2}\left(\gamma_{\ ba}^c + \gamma_{\ cb}^a - \gamma_{\ ac}^b \right).
\end{equation}

This set of equations (\ref{cartan2bis}) is automatically satisfied \cite{Norbert2013} if we take as \textit{governing equations} of the formalism 
the following system:
\begin{enumerate}[i)]
\item Einstein field equations (\ref{EFE}),
\item the Jacobi equations $\left[\partial_{[a},\left[\partial_b,\partial_{c]}\right]\right]=0$, or
\begin{equation}\label{Jacobibis}
\partial_{[a}{{\gamma}^d}_{bc]}-{{\gamma}^d}_{e[a} {{\gamma}^e}_{bc]} =0~,
\end{equation}
\item 18\footnote{3 of which are identities under the Jacobi equations} Ricci equations 
${u^a}_{;d;c}-{u^a}_{;c;d}= {R^a}_{0cd}$ and
\item 20 Bianchi equations $R^a{}_{b[cd;e]} = 0$,
\end{enumerate}
where the $R_{ab}$ components in $(i)$ are replaced, via (\ref{cartan2bis}), in terms of 
commutation coefficients ${{\gamma}^a}_{bc}$ and their derivatives. This system of equations contains a large number of redundancies (e.g. the field equations follow as integrability conditions for the Bianchi equations) and is integrable. For a detailed discussion see \cite{Norbert2013} where the equations have been written out in detail.

\paragraph{Tetrad fixing conventions.} It has become customary \cite{WhiteCollins} to
align $\partial_{3}$ with $\w{\omega}$, such that $\w{\omega} =\omega\partial_{3}\neq 0$.
Applying the commutators $[\partial_3, \partial_\alpha]$ to $p$ and using the Euler and Jacobi equations one can show that the spatial triad can be taken to be co-rotating: $\w{\Omega}+\w{\omega}=0$, with
the remaining tetrad freedom consisting of rotations in the $(1,2)$ plane, $\partial_1+ i\,\partial_2 \to e^{i\alpha} (\partial_1+i\,\partial_2)$ satisfying $\partial_0\alpha =0$.\\
In accordance with the definition of basic variables (see section 3), we will call such transformations \emph{basic rotations}. Notice that, under $\partial_1+i\,\partial_2 \to e^{i\alpha} (\partial_1+ i\,\partial_2)$,
\begin{equation*}
\tfrac{1}{2}(n_{11}-n_{22})+ i \, n_{12} \longrightarrow e^{2 i \alpha} \left(\tfrac{1}{2}(n_{11}-n_{22})+ i \, n_{12}\right),
 \end{equation*}
while under $\sigma_{ab}=0$ and $\w{\Omega}+\w{\omega}=0$ the evolution equations for $n_{11} -n_{22}$ and $n_{12}$ are identical: it therefore follows that one
can specialize the tetrad by means of a basic rotation so as to
achieve $n_{11}=n_{22}\equiv n$. This fixes the tetrad, unless
\begin{equation}\label{extra_rot}
n_{12}=n_{11}-n_{22}=0,
\end{equation}
in which case further basic rotations can (and will) be used to obtain extra simplifications.

\paragraph{Conventions related to the equation of state.} Throughout the paper we assume $p=p(\mu)$ with $p+\mu \neq 0$. We adopt the notations: 
$p' =\ud p / \ud \mu$,  $G\equiv  \frac{p''}{p'}(p+\mu) -p'+\frac{1}{3}$, $G' =\ud G / \ud \mu$ and $G_p=G'/p'$.

\noindent Although the assumption $p+\mu\neq 0$ appears throughout
the literature on the subject, the question whether an \emph{arbitrary} Einstein space can contain a shear-free, but rotating and expanding time-like congruence, seems to have attracted little attention \footnote{See \cite{Pantilie2002} for the analogue question in the Riemannian case. Here the example of the Eguchi-Hanson (Ricci-flat) metric provides us with a shear-free, expanding and rotating congruence.}. As one is setting up a set of 5 partial differential equations for the 3 components of the vector field $\w{u}$, it is clear that some restrictions -- either on the time-like congruence or on the geometry -- seem to be inevitable.  

\begin{re} 
The Ricci equations together with the vanishing of
the shear imply that the magnetic part of the Weyl tensor is determined algebraically by
\begin{eqnarray}
H_{11} = -\omega (\dot{u}_3+r_3), \, H_{22}=-\omega (\dot{u}_3-q_3),\, H_{12}=0, \nonumber \\
H_{13} =  \tfrac{1}{3}z_2 -\omega q_1, \, H_{23}= -\tfrac{1}{3}z_1 + \omega r_2 \label{def_H}.
\end{eqnarray}
\end{re}

For the system of equations yielded by the extended tetrad formalism, imposing the existence of a barotropic equation of state $p=p(\mu)$ as well as the vanishing of the shear results in new chains of integrability conditions. The procedure of building up the sequence of integrability conditions has been
carried out in several papers and for details of their derivation we refer the reader for example to
\cite{NorbertKarimianCarminatiHuf2012}.
The final result of this procedure, taking into account all Jacobi equations and Einstein field equations, the 18 Ricci equations, the contracted Bianchi equations, the `$\dot{\w{E}}$', `$\dot{\w{H}}$' and 
`$\w{\nabla}\cdot \w{E}$ ' Bianchi equations and all integrability conditions on $\mu, \theta, \dot{u}_\alpha$ and $\omega$ (the $[\partial_1,\, \partial_3]\omega$ and $[\partial_2,\, \partial_3]\omega$ relations being equivalent with the `$\w{\nabla}\cdot \w{H}$' equations) is presented in Appendix 1; see also \cite{Norbert2013}, or \cite{MaartensBassett1998} for the compact `1+3 covariant form' of some of these equations.

\section{Formulation in terms of basic variables}
An all-important role in the proof will be played by so-called basic objects (cf.~\cite{Slobodeanu2014} and reference therein), having their origin in the foliation theory.
Let $\HH$ denote the space-like subspace of the tangent space, orthogonal to the velocity $\w{u}$. The component along $\HH$ or the restriction to $\HH$ will
be indicated by a superscript. Recall that a tensorial object $\varsigma$ in $(\otimes^r \HH) \otimes (\otimes^s \HH^*)$ is called \emph{basic} if $(\mathcal{L}_{\w{u}} \varsigma)^\HH =0$, $\mathcal{L}$ denoting here the Lie derivative. In particular,
\begin{de}
A function $f$ on $M$ is \emph{basic} if it is conserved along the flow, $\w{u}(f)=0$, and a vector field $\w{X}$ belonging to $\HH$ is \emph{basic} if $[\w{u}, \w{X}]^\HH=0$.
\end{de}
Some immediate properties of basic functions are provided by the following lemma, the proof of which is easily checked:
\begin{lm} \label{bas}
$(i)$ A linear combination of basic vector fields, with basic coefficient functions, is basic.\\
\medskip
$(ii)$ The horizontal part of the commutator of two basic vector fields is basic.\\
\medskip
$(iii)$ If $\w{X}$ is a basic vector field and $f$ a basic function on $M$, then $\w{X}(f)$ is a basic function on $M$.\\
\end{lm}

In the case of a $\gamma$-law equation of state a length scale $\lambda^{-1}$ was
introduced in \cite{Slobodeanu2014}, enabling one to write pressure and energy density as $p=\frac{r-3}{3}\lambda^r, \mu= \lambda^r$
with $r=3\gamma$. This not only leads to a simplification of the equations, but also plays a key role in some of
the arguments, such as in \emph{Proposition 5} of \cite{Slobodeanu2014}.
In order to generalise this proposition to the case of a general barotropic equation of state and to formulate similar
useful criteria, we will introduce the function $\lambda=\lambda(\mu)$ as follows,
\begin{equation}
 \lambda = \exp \int \frac{\ud \mu}{3(p+\mu)}. 
\end{equation}
The case of a linear equation of state (including a possible non-zero cosmological constant),
$p'=\frac{r}{3}-1$, can then be expressed by
\begin{equation}\label{expl_eqstate}
 p=\left(\frac{r}{3}-1\right) \lambda^r-\mu_0, \ \mu=\lambda^r+\mu_0 , 
\end{equation}
($\mu_0, r$ constants).

Throughout the paper we will assume $p' \notin \left\{0, \pm \tfrac{1}{3}, \tfrac{1}{9}\right\}$ (these cases have been already settled; see the introduction for references).\\

In the next sections we need to identify the basic quantities that recurrently appear in our equations, and that are related to
the variables of the perfect fluid problem. We provide now
the following dictionary\footnote{henceforth \emph{fraktur} symbols will be used to indicate basic objects}, where,
for reasons which will become clear in section 4, we found it convenient to introduce also rescaled acceleration variables $\U_\alpha= \dot{u}_\alpha/(\lambda p^\prime)$:

\begin{lm}\label{dict}
The following modified variables are conserved along the flow (are basic functions):

\begin{eqnarray}
& \O = \frac{p+\mu}{\lambda^{5}}\,\omega,\quad  \N={\frac {n}{\lambda}}, \label{convert_O}\\
& \b_1 = -\smfrac{4}{3}\,{\frac {p+\mu}{{\lambda}^{6}}}z_1 - \smfrac{2}{3}\,{\frac
{\left(9p^\prime - 1 \right) (p+\mu)\,\omega}{{\lambda}^
{5}}}\U_2, \label{convert_b1}\\
& \b_2 = -\smfrac{4}{3}\,\frac {p+\mu}{\lambda^6}z_2 +
\smfrac{2}{3}\,{\frac {\left(9p^\prime - 1 \right) (p+\mu) \omega}{{\lambda}^{5}}}\U_1, \label{convert_b2}\\
& \b_3=-\smfrac{4}{3}\,{\frac {p+\mu}{{\lambda}^{6}}}z_3, \label{convert_b3}\\
& \Q_1 = -\smfrac{1}{3}\,\U_1+ \frac {q_1}{\lambda}, \quad
\R_1=\smfrac{1}{3}\,\U_1 + \frac {r_1}{\lambda}, \label{convert_Q1R1}\\
&\Q_2 = -\smfrac{1}{3}\,\U_2 + \frac{q_2}{\lambda},\quad
\R_2=\smfrac{1}{3}\,\U_2+\frac {r_2}{\lambda}, \label{convert_Q2R2}\\
&{\Q_3}+{\R_3}=-\smfrac{1}{3}\,\U_3 +\frac {q_3}{\lambda}, \quad
{\Q_3}-{\R_3}=\smfrac{1}{3}\,\U_3+\frac{r_3}{\lambda}, \label{convert_Q3R3}\\
& \J=(1-2 G) \dot{\w{U}}^2+\lambda^{-2}\left(\theta^2-3 \mu-2\frac{j}{p'}\right)+9\O^2\frac{\lambda^8}{(p+\mu)^2}, \label{convert_J}\\
& \E_{\alpha \beta} = \frac{3p' + 1}{\lambda^2 p'}E_{\alpha \beta} + G \U_\alpha\U_\beta, \quad (\alpha,\beta)=(1,2),\,
(1,3),\, (2,3), \label{convert_E123}\\
& \E_0 = \frac{3p' + 1}{\lambda^2 p'}(E_{11}-E_{22}) + G (\U_1^2-\U_2^2) ,\label{convert_E0}\\
& \E_3 = \frac{3p' + 1}{\lambda^2 p'}E_{33}-\frac{G}{3} (\U_1^2+\U_2^2-2\U_3^2)+\frac{2(9p'+1)\O^2}{3p'}\frac{\lambda^8}{(p+\mu)^2} .\label{convert_E3}
\end{eqnarray}
\end{lm}

\begin{proof}
Straightforward but lengthy computation using the propagation rules (see appendix 1) for each quantity involved.
\end{proof}

\begin{re}\label{Einvar}
Note that the basic objects $\smfrac{1}{2}\E_0 + i \E_{12}$ and $\E_{13} + i \E_{23}$, transform as follows under a basic rotation in the $(1,2)$ plane:
\begin{eqnarray*}
\smfrac{1}{2}\E_0 + i \E_{12} \longrightarrow e^{2 i \alpha} (\smfrac{1}{2}\E_0 + i \E_{12}), \nonumber \\
\E_{13}+i \E_{23} \longrightarrow e^{i \alpha} (\E_{13}+i \E_{23}).
\end{eqnarray*}
This shows that conditions like $\E_{13}=\E_{23}=0$, occurring in for example Lemma \ref{Theorem3}, have a truly invariant (frame-independent) 
meaning.
\end{re}

It will be convenient also to rewrite the spatial basis in terms of the basic vector fields 
\begin{equation}\label{basxyz}
\X=\lambda^{-1}\partial_1, \quad \Y = \lambda^{-1}\partial_2, \quad \Z = \lambda^{-1}\partial_3.
\end{equation} 
It follows then that
$\U_1= -3 \X(\ln \lambda)$, $\U_2= -3 \Y(\ln \lambda)$, and $\U_3= -3 \Z(\ln \lambda)$. 

Acting with the operators (\ref{basxyz}) on (\ref{convert_O}$_a$) one obtains 
the basic equations

\begin{equation}\label{XYZO}
\X(\O) = -\O \Q_1 - \half \b_2, \quad
\Y(\O) =  \O \R_2 + \half \b_1, \quad
\Z(\O) = -2 \O \R_3 . 
\end{equation}

with integrability conditions given by (\ref{basic_eq1}+\ref{basic_eq6},\ref{basic_eq2},\ref{basic_eq3}).

\begin{re} 
In terms of the vector fields (\ref{basxyz}), many quantities in Lemma \ref{dict} are easily recognised as being basic. First recall \cite{Slobodeanu2014} that to our fluid one can locally associate a transversally conformal submersion $\varphi:(M,g) \to (N,h)$ onto a Riemannian $3$-manifold having $\w{u}$ tangent to its fibres. Then notice that any tensorial quantity constructed by pull-back is clearly basic. In particular, since $\lambda^2 g^\mathcal{H} = \varphi^*h$ is basic, it follows (using also Lemma \ref{bas}) that $\N = \lambda^2 g([\Z, \X], \Y)$ is basic. A similar argument holds for $\Q_i$ and $\R_i$:
\begin{eqnarray*}
\fl \Q_{1} &= \lambda^2 g([\Z, \X], \Z) , \quad \Q_{2}= \lambda^2 g([\X, \Y], \X), 
\quad \Q_{3} = \tfrac{\lambda^2}{2}\big( g([\Z, \X], \X)-g([\Z, \Y], \Y)\big), \\
\fl  \R_{1} & = \lambda^2 g([\X, \Y], \Y), \quad \R_{2}= -\lambda^2 g([\Z, \Y], \Z), \  \R_{3}=-\tfrac{\lambda^2}{2}\big( g([\Z, \X], \X) + g([\Z, \Y], \Y)\big).
\end{eqnarray*}
As to $\J$, $\E_{\alpha \beta}$, $\E_{0}$ and $\E_{3}$, they correspond, up to constant factors, to the following pull-backed curvatures of the `material manifold' $N$: $R^N \circ \varphi$ (scalar curvature), $\varphi^* R_{\alpha \beta}^N$ 
($\alpha \neq \beta$), $\varphi^* R_{11}^N -\varphi^* R_{22}^N$, and $\frac{1}{3}R^N \circ \varphi-\varphi^* R_{33}^N$ (Ricci curvatures), respectively; in particular, they 
are basic functions. Finally, if the equation of state is linear, $\O$, $\b_1$, $\b_2$ and $\b_3$ correspond respectively to the basic functions 
$\frac{p'+1}{2}\Omega(\X, \Y)$, $(p'+1)\beta(\X)$, $(p'+1)\beta(\Y)$ and $(p'+1)\beta(\Z)$, defined in \cite{Slobodeanu2014}.
\end{re}

Translating the equations of Appendix 1 in terms of the basic variables $\O,\J,\b_{\alpha},\Q_\alpha,\R_\alpha,\E_0,\E_3$
and the non-basic variables $p, \mu, \theta, \U_{\alpha}$, augmented with all information obtainable by acting with the $\partial_\alpha$ operators 
on the remaining dictionary elements, one derives a set of equations which can be split
into
\begin{itemize}
\item evolution equations for the non-basic quantities $\mu, \theta, \U_\alpha$, namely (\ref{e0mu}) and
\begin{eqnarray}
 \partial_0 \theta = \half \lambda^2p'[(1-2 G) \dot{\w{U}}^2-\J]+\half \O^2(9p'+4)(p+\mu)^{-2}\lambda^{10}\nonumber \\
\ \ \ \ \ \ \ +\tfrac{1}{6}(3p'-2)\theta^2- \tfrac{1}{2}(3p'+1)\mu -\tfrac{3}{2}p, \label{utheta}\\
 \partial_0 \U_1 = p' \theta \U_1+\frac{\lambda^5}{p+\mu}\left(-\frac{3}{4}\b_1-\frac{9p'-1}{2}\O  \U_2 \right) ,\label{uU1}\\
 \partial_0 \U_2 = p' \theta \U_2+\frac{\lambda^5}{p+\mu}\left(-\frac{3}{4}\b_2+\frac{9p'-1}{2}\O \U_1\right) ,\label{uU2}\\
 \partial_0 \U_3 = p' \theta \U_3 -\frac{\lambda^5}{p+\mu}\frac{3\b_3}{4},\label{uU3}
\end{eqnarray}

\item purely basic equations, which will play only a minor role and which, for the sake of readability, are presented in Appendix 2,

\item algebraic equations in $\dot{U}_\alpha$ and $\theta$, with the basic functions
$\B_1, \ldots, \B_{16}$ defined by equations
(\ref{basic_eq11}--\ref{basic_eq26}) of Appendix 2,
\end{itemize}
\begin{eqnarray}
\fl \ \ \ \ \ \ \diamondsuit \ \textrm{the } \dot{\w{H}} \textrm{ equations (\ref{H13_0},\ref{H23_0},\ref{H12_0},\ref{H11_0},\ref{H22_0}):} \nonumber \\
\fl \left[\frac{3p'^2}{3p'+1}\E_{12}-\h_8\O\frac{\lambda^3 \theta}{p+\mu} \right]\U_1-\smfrac{3}{2}\left[\frac{p'^2}{3p'+1}(\E_0-3\E_3)+
\frac{(p'+1)(81p'^2-5)}{2(3p'+1)}\frac{\O^2\lambda^8}{(p+\mu)^2}\right]\U_2\nonumber \\
-\frac{3p'^2}{3p'+1}\E_{23}\U_3 -(3p'+1)\big(\Q_1\O + \tfrac{1}{4} \b_2 \big)\frac{\lambda^3 \theta}{p+\mu}\nonumber \\
- 3\O\big(\R_2\O+\tfrac{1}{8}(9p'+11)\b_1\big)\frac{\lambda^8}{(p+\mu)^2}
+\B_{12}=0, \label{m_eq12}\\[3mm]
\fl -\tfrac{3}{2}\left[\frac{p'^2}{3p'+1}(\E_0+3\E_3)-
\frac{(p'+1)(81p'^2-5)}{2(3p'+1)}\frac{\O^2\lambda^8}{(p+\mu)^2}\right]\U_1 - \left[\frac{3p'^2}{3p'+1}\E_{12}+\h_8\O \frac{\lambda^3 \theta}{p+\mu} \right]\U_2\nonumber \\
+\frac{3p'^2}{3p'+1}\E_{13}\U_3 +(3p'+1)\big(\R_2\O + \tfrac{1}{4}\b_1 \big) \frac{\lambda^3 \theta}{p+\mu}\nonumber \\
- 3\O\big(\Q_1\O+\tfrac{1}{8}(9p'+11)\b_2\big)\frac{\lambda^8}{(p+\mu)^2}+\B_{13}=0 \label{m_eq13},\\[3mm]
\fl \frac{3{p'}^2}{3{p'}+1}(\E_{13}\U_1-\E_{23}\U_2-\E_0\U_3)+ 6\Q_3\frac{\O^2\lambda^8}{(p+\mu)^2}-\B_{11}=0,  \label{m_eq11}\\[3mm]
\fl \frac{6 {p'}^2}{3{p'}+1}(\E_{23}\U_1+\E_{13} \U_2-2 \E_{12}\U_3)- 2(3p'+1)\Q_3\O\frac{\lambda^3 \theta}{p+\mu} + \B_9-\B_{10}
= 0,\label{m_eq9}\\[3mm]
\fl \frac{6 {p'}^2}{3{p'}+1}(\E_{23}\U_1-\E_{13} \U_2)-\frac{\lambda^3\theta \O}{p+\mu}\left(\tfrac{2}{3}(9{p'}G-9{p'}^2+1) \U_3
+2(3{p'}+1)\R_3 \right)\nonumber \\
 -\tfrac{9}{2}(p'+1)\b_3\frac{\O\lambda^8}{(p+\mu)^2} -\B_9-\B_{10} = 0,\label{m_eq10}\\[3mm]
\fl \ \ \ \ \ \ \diamondsuit \ \textrm{the integrability conditions } \partial_A \partial_3 \theta - \partial_3 \partial_A \theta-[\partial_A,\, \partial_3] \theta
=  0\ (A=1,2): \nonumber \\
\fl \h_1 \O \U_2\U_3 + \tfrac{5}{2}(3p'-1) \b_3 \U_1 + 2(9{p'}-1)\O(\Q_3+\R_3)\U_2 \nonumber \\
+\left[2 \O \R_2(9{p'}-1)-\b_1(3{p'}-2)\right]\U_3 - \frac{2(9{p'}-1)}{3{p'}+1}\O\E_{23} - 3 \B_1 = 0, \label{m_eq1} \\[3mm]
\fl  \h_1 \O \U_1\U_3 -\tfrac{5}{2}(3{p'}-1)\b_3 \U_2 -2(9{p'}-1)\O(\Q_3-\R_3)\U_1
\nonumber \\
 -\left[2 \O \Q_1(9{p'}-1)-\b_2(3{p'}-2)\right]\U_3 - \frac{2(9{p'}-1)}{3{p'}+1}\O\E_{13}+3\B_2 = 0, \label{m_eq2} \\[3mm]
\fl \ \ \ \ \ \ \diamondsuit \ \textrm{equations resulting by evaluation of } \partial_1(\ref{convert_b2})-\partial_2(\ref{convert_b1}) \textrm{ and }
\partial_2(\ref{convert_b2})-\partial_1(\ref{convert_b1}):\nonumber \\
\fl \O [ \h_2 (\U_1^2+\U_2^2)+\h_3 \U_3^2] + 3 (3  {p'} +1) (9{p'}-1) \O  \R_3 \U_3+\smfrac{3}{2}
(3  {p'} +1) (15  {p'}  \mu+6 p+\mu)\lambda^{-2}\O\nonumber \\
+\smfrac{3}{4}(3  {p'} +1)  [2 \O  \Q_1(9{p'} -1)+3 \b_2(4{p'} -1)] \U_1 \nonumber \\
- \smfrac{3}{4}(3  {p'} +1) [2 \O  \R_2(9{p'} -1)+3 \b_1(4{p'} -1)] \U_2  \nonumber \\
-\tfrac{5}{2}(3p' + 1)(3p' - 1)\O \lambda^{-2} \theta^2 -\smfrac{3}{2}  (135  {p'}^2 + 96 p' + 1) \lambda^8 (p+\mu)^{-2}
\O ^3 \nonumber \\
+[\half(45{p'}^2 + 12 p' -1) \J
+\smfrac{3}{2} (9  {p'} -1) \E_{3}] \O + \tfrac{9}{4}(3p' + 1)\B_3  = 0, \label{m_eq3} \\[3mm]
\fl \big[\h_4 (\U_1^2 - \U_2^2) - p'(3G - 2) \E_0\big]\theta + \frac{(3p' +1)\lambda^5}{p+\mu}\big\{\h_5 \O \U_1\U_2 -(9p' -1)  \O \E_{12}
\nonumber \\
+\tfrac{1}{4}
[\h_6 \b_1-2 (3  p' +1) (9  p'-1 ) \O  \R_2] \U_1 -\tfrac{1}{4}[\h_6 \b_2-2 (3  p' +1) (9  p' - 1) \O  \Q_1] \U_2 \nonumber \\
+\tfrac{3}{4} (3  {p'} +1) \B_4 \big\} = 0, \label{m_eq4} \\[3mm]
\fl \ \ \ \ \ \ \diamondsuit \ \textrm{equations resulting by evaluation of } \partial_1(\ref{convert_b2}) + \partial_2(\ref{convert_b1}) \textrm{ and }
\partial_2(\ref{convert_b2}) + \partial_1(\ref{convert_b1}): \nonumber \\
\fl 2\big[- \h_4  \U_1\U_2 + p'(3 G-2) \E_{12}\big]\theta
+ \frac{(3p' +1)\lambda^5}{p+\mu}\big\{\tfrac{1}{2}\h_5 \O (\U_1^2-\U_2^2) - 
\tfrac{1}{2}(9 p' - 1)\O  \E_0 + \tfrac{3}{4} (3 p' + 1) \B_5 \nonumber \\
-\tfrac{1}{4}[\h_6 \b_2 - 2(3p' + 1)(9p' -1) \O  \Q_1] \U_1 -\tfrac{1}{4}[\h_6 \b_1-2 (3  {p'} +1) (9  {p'} -1 ) \O  \R_2] \U_2 \nonumber\\
-(9  {p'} -1) (3p' +1) \O  \Q_3 \U_3 \big\}  = 0, \label{m_eq5} \\[3mm]
\fl \big[\h_4 (\U_1^2 + \U_2^2 - 2 \U_3^2) + 3p'(3G - 2) \E_3 \big] \theta + \frac{(3p' +1)\lambda^5}{p+\mu}\big\{-\tfrac{9}{4} (3p' +1) \B_6 \nonumber \\
\fl \ \ \ \ + \tfrac{3}{4} [2 (3  {p'} +1) (9  {p'} -1) \O  \R_2+(6 G  {p'} +3  {p'} +1) \b_1]\U_1 \nonumber \\
\fl \ \ \ \ + \tfrac{3}{4} [2 (3  {p'} +1) (9  {p'} -1) \O  \Q_1+(6 G  {p'} +3  {p'} +1) \b_2]\U_2 \nonumber\\
\fl \ \ \ \ -\tfrac{9}{4} (4 G  {p'} -9  {p'} ^2+1) \b_3 \U_3 \big\}
-\frac{\lambda^8 \theta}{(p+\mu)^2}(36 G  {p'} -81  {p'} ^3+27  {p'} ^2-27  {p'} -7) \O ^2 = 0, \label{m_eq6} \\[3mm]
\fl \ \ \ \ \ \ \diamondsuit \ \textrm{equations resulting by evaluation of } \partial_1(\ref{convert_b3}) \textrm{ and } \partial_2(\ref{convert_b3}):\nonumber \\
\fl \big[\h_4  \U_1\U_3 - p'(3 G-2)\E_{13}\big]\theta + \frac{(3p' +1)\lambda^5}{p+\mu} \big\{\h_7 \O \U_2\U_3
+\tfrac{1}{8}(18 G  {p'} -63  {p'} ^2+7)\b_3 \U_1
\nonumber \\
-\smfrac{1}{2} (9  {p'} -1) (3  {p'} +1) (\Q_3-\R_3) \O  \U_2 \nonumber \\
+ \tfrac{1}{4}(9Gp' -9{p'}^2+1) \b_1 \U_3 - \smfrac{3}{4}(3p'+1)\B_7\big\} = 0,  \label{m_eq7} \\[3mm]
\fl \big[\h_4  \U_2 \U_3 - p'(3 G-2) \E_{23}\big]\theta + \frac{(3p' +1)\lambda^5}{p+\mu} \big\{- \h_7 \O \U_1\U_3
+\tfrac{1}{8}(18 G  {p'} -63  {p'} ^2+7) \b_3 \U_2 \nonumber \\
-\smfrac{1}{2} (9  {p'} -1) (3  {p'} +1) (\Q_3+\R_3) \O  \U_1 \nonumber \\
+ \tfrac{1}{4}(9 G  {p'} -9  {p'} ^2+1) \b_2 \U_3 -\smfrac{3}{4} (3  {p'} +1) \B_8 \big\}= 0. \label{m_eq8}
\end{eqnarray}

In these equations $\h_1,\ldots, \h_7$ are functions of $\mu$ defined by
\begin{eqnarray}
\h_1 = 2\frac{36 {p'} G +9{p'}^2-1}{3(3{p'}+1)},\nonumber\\
\h_2 = 45 G {p'}^2+21 G {p'} +9 {p'}^2-1,\nonumber\\
\h_3= \half (90 G  {p'} ^2+6 G  {p'} +9  {p'} ^2-1),\\
\h_4 = p'\big(3  G'  (p+\mu)(1+3{p'})+3 G^2 -18 G  {p'} ^2-6 G  {p'} -2 G\big),\nonumber\\
\h_5 =  6 Gp'(9p' + 1) - 54{p'}^3 + 9{p'}^2 + 6  {p'} -1,\nonumber\\
\h_6 = 18 G p'-54 p'^2-3 p'+5,\nonumber\\
\h_7 = \tfrac{1}{6}(9p' -1)(9 G p' - 9{p'} ^2+1),\nonumber \\
\h_8 = \half (9Gp'-9p'^2+1) \nonumber .
\end{eqnarray}

\noindent Note that (\ref{m_eq1}, \ref{m_eq2}, \ref{m_eq3}, \ref{m_eq4}) in the case of a linear equation of state correspond respectively to equations (26), (27), (28) and (23) of \cite{Slobodeanu2014}.

\section{General theorems}

In this section we present some criteria which will be used later on, but which also may turn out
to be helpful when tackling the conjecture for a general barotropic equation of state. We begin with two theorems generalising Proposition 5 of
\cite{Slobodeanu2014} for arbitrary $p(\mu)$.

\begin{te} \label{th1}
If for a shear-free perfect fluid obeying a barotropic equation of state $\U_1$ and $\U_2$ are basic,
then $\omega \theta = 0$.
\end{te}

\begin{proof} Assume $\omega \theta \neq 0$. Using the evolution equations for $\mu, \u_1,\u_2$, the conditions $\partial_0(\U_1)=
\partial_0(\U_2)=0$ are equivalent with  $z_1+\theta \u_1 = z_2+\theta \u_2=0$. Applying $\partial_0$ to the latter two equations and substituting for $z_1,z_2$
yields a homogeneous system in $\u_1,\u_2$, the coefficient matrix of which is positive definite (in which case \cite{WhiteCollins} applies and the proof ends), unless
\begin{equation}\label{th1_eq1}
 \left(\frac{2}{3}-G-2p'\right)\theta^2+2\omega^2-\frac{\mu+3 p}{2}+j = 0
\end{equation}
and
\begin{equation}\label{th1_eq2}
G + p^\prime - \frac{1}{3} = 0.
\end{equation}

The second of these conditions implies that we have a linear equation of state ($p''=0$), which, when substituted in the
first, yields
\begin{equation}\label{th1_j}
 j= \left(p'-\frac{1}{3}\right) {\theta}^{2}-2\,{\omega}^{2}+\frac{\mu + 3p}{2}.
\end{equation}
Acting on this with $\partial_3$ gives, using (\ref{e3j}), $\theta (z_3+\theta \u_3)(3p'-2)=0$. If $z_3+\theta \u_3=0$ then, by (\ref{gradmu}, \ref{gradtheta}),
we have $z_\alpha -\theta \frac{\partial_\alpha p}{p+\mu} =0$ and hence, with $F=\log \theta - \int (p+\mu)^{-1}\ud p$,
$\ud F = -\partial_0 F \w{u}^\flat$. This shows that $\w{u}$ is hypersurface orthogonal (and hence the vorticity vanishes), unless $F$ is constant, whence $\theta=\theta(\mu)$, which is the case treated in \cite{TreciokasEllis,Langth, Lang}.
Hence, as $\omega \theta \neq 0$, we necessarily have $p'=2/3$.
Then however the $\partial_1$ derivative of $z_2+\theta \u_2=0$ gives $2\omega^2 \theta^2-z_3^2-\theta \u_3 z_3=0$, which can be used to eliminate the
$\u_3 z_3$ term arising in $\partial_0(\ref{th1_j})$ to yield  $\theta^2(20 \omega^2+5 \u_3^2+\frac{9}{2}(p+\mu))-5z_3^2=0$. Propagating this again along $\w{u}$ and using the previous results to eliminate $z_\alpha$ and $j$, eventually gives
$(p+\mu) \theta^2=0$.
\end{proof}

\begin{te} \label{th2}
If for a rotating and expanding shear-free perfect fluid, obeying a barotropic equation of state, $\U_3$ is basic, then a Killing vector along the vorticity exists.
\end{te}

\begin{proof} As in Theorem \ref{th1} one sees that $\partial_0(\U_3)=0$ is equivalent with $z_3+\theta \u_3=0$.
Propagating this along $\w{u}$ shows, using the evolution equations of Appendix 1, that $\u_3=0$ or (\ref{th1_eq1}) holds.\\

We first show that $\u_3\neq 0$ is inconsistent with $\omega \theta \neq 0$. 

Acting on (\ref{th1_eq1}) with $\partial_3$, one obtains in place of (\ref{th1_eq2}),
\begin{equation}\label{th1_eq3}
 G+p'-\frac{1}{3} +\frac{p+\mu}{3}G_p= 0.
\end{equation}
Using this to eliminate $G_p$ from the relations obtained by acting with $\partial_1$ and $\partial_2$ on (\ref{th1_eq1}),
one finds
\begin{eqnarray}
\fl  & & {\theta}^{2} \left( 6G+9{p'}-4 \right) {\u_1}-\tfrac{3}{2}\omega
\theta\left( 9G-2 \right) {\u_2}+\theta\left( 6G+9{
p'}-4 \right) {z_1}-\tfrac{3}{2}\omega\left( 1-9{p'} \right) {
z_2}=0 , \label{th1_eq4}\\
\fl & & \tfrac{3}{2}\omega\theta\left( 9G-2 \right) {\u_1}+{\theta}^{2}
 \left( 6G+9{p'}-4 \right) {\u_2}+\tfrac{3}{2}\omega\left( 1-9{
p'} \right) {z_1}+\theta\left( 6G+9{p'}-4
 \right) {z_2}=0. \label{th1_eq5}
\end{eqnarray}
In the case of a linear equation of state ($G+p'-\frac{1}{3}=0$) this becomes a homogeneous system
in the variables $z_1+\theta \u_1, z_2+\theta \u_2$, with a coefficient matrix which is
positive definite, unless $\theta(p'-2/3)=\omega(9p'-1)=0$ and hence we are done by Theorem \ref{th1}.
If there is no linear equation of state (in which case the $(\u_1, \u_2)$-coefficient matrix of (\ref{th1_eq4}, \ref{th1_eq5}) is positive definite), solving (\ref{th1_eq4}, \ref{th1_eq5}) for $\u_1,\u_2$ leads to expressions which are homogeneous in $z_1,z_2$.
Propagating (\ref{th1_eq4}, \ref{th1_eq5}) along $\w{u}$, one obtains a new homogeneous system $a z_1+b z_2=-b z_1+a z_2 =0$ with
\begin{eqnarray}
\fl a &=& -36\omega \theta (3G+3p'-1)^2(81 G^2+162G p'-96 G-72p'+28), \label{th1_eq6}\\
 \fl b &=& 6(3G+3p'-1)[4\theta^2(6G+9p'-4)(27 G^2-33G+81 G p'+54 p'^2-45 p'+10) \nonumber\\
\fl  && + 9\omega^2(9G-2)(3G+12p'-12).
\end{eqnarray}
Again the coefficient matrix is positive definite, as $a^2+b^2 = \partial_0 b = 0$ would lead
to an inconsistency with (\ref{th1_eq3}), unless we have a linear
equation of state. It follows that $z_1=z_2=0$, hence also $\u_1=\u_2=0$ and we are done by Theorem \ref{th1}.\\

Having excluded the case $\u_3\neq 0$, we now turn to the case where acceleration and vorticity are orthogonal:
$\u_3=0$ and hence also $z_3=0$. From the $\partial_\alpha \u_3=0$ equations one now obtains
\begin{eqnarray}
 E_{13} - r_3 \u_1 = 0,\label{th1_eq7}\\
 E_{23} +q_3 \u_2 = 0,\label{th1_eq8}\\
 E_{33}+\tfrac{1}{3}j+\tfrac{2}{3}\omega^2-\u_1q_1+\u_2r_2=0 , \label{th1_eq8bis}
\end{eqnarray}
with $\partial_0$(\ref{th1_eq7},\ref{th1_eq8}) leading to two further equations,
\begin{eqnarray}
 p' r_3[(3G-2)\theta \u_1-(3p'+1)z_1)]=0,\label{th1_eq9}\\
 p' q_3[(3G-2)\theta \u_2-(3p'+1)z_2)]=0. \label{th1_eq10}
\end{eqnarray}
Clearly $r_3q_3 \neq 0$ implies the existence of a function $f(\mu)$, such that $\partial_\alpha(f(\mu) \theta)=0$ and hence either $\w{u}$ is hypersurface orthogonal ($\omega=0$) or $\theta=\theta(\mu)$ (and then again $\omega \theta =0$).

It follows that we can restrict to the cases $r_3\neq 0=q_3$ or $r_3=0\neq q_3$ (which are equivalent under a discrete rotation) and the case $q_3=r_3=0$.
The latter is easy: by (\ref{th1_eq7},\ref{th1_eq8}) we have $E_{13}=E_{23}=0$, with $j$ given by (\ref{th1_eq8bis}).

One can verify that herewith the $\partial_3$ derivatives of all invariants vanish, implying the existence of a Killing
vector $K \partial_3$ along the vorticity.
Alternatively one can explicitly verify the existence of this Killing
vector, by showing that the Killing equations $k_{(a;b)}=0$, with
$\w{k} = K \partial_3$, form an integrable set. The Killing equations are given in explicit form by
\begin{equation}
\partial_0 K-\frac{\theta}{3}K = \partial_1 K - q_1 K =
\partial_2 K + r_2 K = \partial_3 K = 0
\end{equation}
and acting on $K$ with the commutators (\ref{comm_explicit}) shows that the resulting integrability conditions are identically satisfied:
except for $[\partial_1,\, \partial_2] K$, this is an immediate consequence of the equations (\ref{e0r},\ref{e0q},\ref{divH1},\ref{divH2}), while
for $[\partial_1,\, \partial_2] K$ one has to use equations (\ref{e2q1},\ref{e1r2}).

It remains to show the inconsistency of (for example) the case $r_3\neq 0=q_3$: taking a $\partial_3$ derivative of (\ref{th1_eq9})
(with $\partial_3 G =0$) and expressing that
$\partial_1 q_3=0$, one obtains $n=0$ and $\u_1-r_1=0$. In terms of the basic variables introduced in section 2 we have then the following restrictions:
\begin{eqnarray}
\N=\b_3 = \Q_3+\R_3=\E_{23}=\E_{13}+6 \R_1 \R_3=0, \\
\U_1= \frac{3 \R_1}{1+3p'},
\end{eqnarray}
with $\R_3\neq 0$. Herewith the algebraic equation (\ref{m_eq2}) can be written as
\begin{equation}
 8 \O \R_1 \R_3 (9p'-1) +\B_2(3p'+1)=0, \label{th1_eq15}
\end{equation}
implying that $p'$ is basic (and hence $\theta=0$, unless $p'$ is constant), or that $\R_1=\B_2=0$. In the latter case
 (\ref{m_eq11}) reads $6\O^2 \R_3 \lambda^8+\B_{11}(p+\mu)^2=0$, implying that $\lambda^{8}(p+\mu)^{-2}$ is a basic function and hence $\theta=0$. On the other hand, in the case when $p'$ is constant (linear equation of state) and $\R_1 \neq 0$, equation (\ref{m_eq11}) reduces to
\begin{equation}
6\O^2 \R_3  \frac{\lambda^8}{(p+\mu)^2} + \B_{11} +\frac{54 p'^2}{(3p'+1)^2}\R_1^2\R_3=0
\end{equation}
and again we see that $\lambda^{8}(p+\mu)^{-2}$ is basic, so $\theta=0$.
\end{proof}

As the previous theorem applies in particular to the case $\U_3=0$ ($\dot{\w{u}}$ orthogonal to $\w{\omega}$) and as, vice versa, the existence
of a Killing vector along the vorticity automatically implies $\U_3=0$, we also obtain the following corollary:
\begin{co}\label{co1}
For a rotating and expanding shear-free perfect fluid with a barotropic equation of state, the acceleration is orthogonal to the vorticity if and only if a Killing vector exists along the vorticity.
\end{co}

\section{Linear equation of state}

We first demonstrate in Lemma \ref{Theorem3} that for a linear equation of state the vanishing of the basic variables $\E_{13},\E_{23}$ and $\Q_3$ implies the existence
of a Killing vector along the vorticity. In Theorem \ref{Theorem4} we show that for a linear equation of state the conjecture holds true, unless
the conditions for Lemma \ref{Theorem3} are satisfied.
The final `elusive case' \cite{CollinsKV}, in which there is
a Killing vector along the vorticity, is then dealt with in Theorem \ref{Theorem5}.

Recall first an observation which will be helpful in the sequel.
\begin{re}[\cite{Slobodeanu2014}]\label{pol}
If a function $f$ on $M$ satisfies $\alpha_n f^n + ... + \alpha_1 f + \alpha_0 =0$, where $n \in \mathbb{N}$ and $\alpha_i$'s are all basic functions, then either $f$ is basic or $\alpha_i=0$ for all $i=0,1, ..., n$.
\end{re}

\begin{lm}\label{Theorem3}
If for a rotating and expanding shear-free perfect fluid, obeying a linear equation of state 
$p=(\gamma-1)\mu + constant$, the basic variables $\E_{13},\E_{23}$ and $\Q_3$ vanish, then a Killing vector exists along the vorticity.
\end{lm}

\begin{proof}
Let us assume that $\U_3$ is not basic, as otherwise Theorem \ref{th2} applies. Since the equation of state is linear, the determinant of the linear system (\ref{m_eq1},\ref{m_eq2}) in $\U_1$, $\U_2$
\begin{equation}\label{dete}
D=\O^2(9p'-1)^2\left[(3p'-1)^2 \U_3^2 -6 \R_3 (9p'^2-1) \U_3 + \textrm{basic terms}\right],
\end{equation}
has basic coefficients, so can be assumed to be  non-zero due to Remark \ref{pol}. Solving this system for $\U_1,\U_2$ yields rational expressions in $\U_3$ with basic coefficients; this allows us in the following to obtain various equations only in terms of $\U_3$.

Since by hypothesis $\Q_3=0$  we have $n_{12}=n_{11}-n_{22}=0$  and we are free to choose a basic rotation making for example $\E_{12}=0$.
Together with $\E_{13}=\E_{23}=0$ and the conditions for a linear equation of state ($G'=G+p'-\smfrac{1}{3}=0$) the equations of the previous section
simplify considerably. In particular one obtains from (\ref{m_eq11}, \ref{m_eq9}) $\B_{10}=\B_9$ and
$ \frac{3p'^2}{3p'+1}\E_0 \U_3 + \B_{11}=0$.
By Theorem \ref{th2} this implies the existence of a Killing vector along the vorticity, unless $\E_0=\B_{11}=0$. Since $\E_{12}$ and $\E_0$ are both zero, a further basic rotation may be taken (cf. Remark \ref{Einvar}), making $\b_1 = \b_2$.

By (\ref{basic_eq24}) we have then $\B_9=\B_{10}= -3 \b_3\E_3 / (8 \O)$,
such that (\ref{m_eq10}) simplifies to
\begin{equation}\label{specsimp}
\fl \ \frac{\lambda^3 \theta}{p+\mu}\left((6 p^\prime + 1)(3 p' - 1)\U_3 - 3 (3\,{p'}+1 ) {\R_3}\right)-
\frac{27\lambda^8}{4(p+\mu)^2}(p'+1) \b_3 + \frac{9\E_3 \b_3}{8\O^2}=0
\end{equation}
showing that $\b_3 \neq 0$ unless $p'=-1/6$.

\medskip
We see Equations (\ref{m_eq6},\ref{m_eq7},\ref{specsimp}) as an algebraic system in the variables $\lambda^3 \theta(p+\mu)^{-1}$ and $\lambda^8 (p+\mu)^{-2}$; by eliminating 
the first variable from (\ref{specsimp}) and (\ref{m_eq7}), then from (\ref{specsimp}) and (\ref{m_eq6}), and finally taking the 
resultant of the two relations with respect to the second variable, we obtain a compatibility condition in the form of a polynomial equation 
in $\U_3$ with basic coefficients. But then Remark \ref{pol} requires that the leading coefficient is vanishing; assuming $p' \neq -1/6$ this is 
equivalent to
$$
\R_2 = \frac{9{p^\prime}^2-6p^\prime - 1}{2(3p^\prime+1)(9p^\prime-1)}\frac{\b_1}{\O}.
$$
Analogously, repeating the argument with (\ref{m_eq8}) in the place of (\ref{m_eq7}), we get $\R_2 = \Q_1$.

By propagating the equations (\ref{m_eq1},\ref{m_eq2}) we obtain a homogeneous system in $\theta$ and $\lambda^5 (p+\mu)^{-1}$, whose necessarily vanishing determinant leads us to a third degree polynomial equation in $\U_3$ with basic coefficients. Again (cf. Remark \ref{pol}) this requires the cancellation of every coefficient. This shows us that $\b_1=\b_2=0$ would imply  $\U_1=\U_2=0$ (a contradiction, cf. Theorem \ref{th1}), so we may assume $\b_1 \neq 0$. The vanishing of the leading coefficient yields a formula for $\R_3$, which we substitute in the second degree coefficient, from which we obtain two possible expressions of $\b_3$ (in terms of other basic quantities), unless $5p^\prime + 1=0$.

If $5p^\prime + 1 \neq 0$, then substituting each of the two expressions of $\b_3$ (and after taking into account further conditions 
arising from the cancellation of lower degree coefficients of the basic polynomial) shows 
that $\U_1, \U_2$ are basic, so Theorem \ref{th1} applies.

If $5p^\prime + 1 = 0$, then $\partial_0$(\ref{m_eq1}) and (\ref{m_eq7}) form another homogeneous system in $\theta$ and $\lambda^5 (p+\mu)^{-1}$, the necessarily vanishing determinant of which leads us to a new polynomial equation in $\U_3$ with basic coefficients for which Remark \ref{pol} applies. Again by cumulating step by step the constraints issued from the cancellation of various coefficients, we are led finally to the same outcome: $\U_1, \U_2$ should be basic and Theorem \ref{th1} applies.

\medskip
When $p'=-1/6$ the above formulae for $\R_2$ and $\Q_1$ no longer hold (so neither do the subsequent considerations), but now equations (\ref{m_eq4},\, \ref{m_eq5}) reduce to
\begin{eqnarray}
 (\O \R_2+\b_1)\U_1-(\O\Q_1+\b_2) \U_2-2 \O \U_1\U_2+\smfrac{3}{5}\B_4=0,\\
 \O(\U_2^2-\U_1^2)-(\O\Q_1+\b_2)\U_1-(\O\R_2+\b_1)\U_2+\smfrac{3}{5}\B_5=0
\end{eqnarray}
and elimination of $\U_1$ or $\U_2$ results in a fourth degree polynomial relation for $\U_2$ or $\U_1$, with basic coefficients and leading coefficient $\O^3$. It follows that $\U_1$ and $\U_2$ are basic and we are done by Theorem \ref{th1}.
\end{proof}

Full details of the previous proof can be found in Maple or Mathematica worksheets, which are available from the authors.

The $p'\neq -1/6$ part of the above proof follows closely the $\widetilde{b}_3=\widetilde{c}_3=0$ case in the proof given in Section 5 of \cite{Slobodeanu2014}, which is independent of whether the cosmological constant vanishes or not and which left aside the exceptional cases $p' \in \{-1/6, -1/5\}$.

\begin{te}\label{Theorem4}
 If a rotating and expanding shear-free perfect fluid obeys a linear equation of state $p=(\gamma-1)\mu + constant$, then a Killing vector exists along the vorticity.
\end{te}

\begin{proof}
For a linear equation of state the determinant of the linear system (\ref{m_eq1},\ref{m_eq2}) in $\U_1$, $\U_2$ is given by (\ref{dete}) and hence can be assumed to be  non-zero, unless $\U_3$ is basic (in which case Theorem \ref{th2} applies). Solving this system for $\U_1,\U_2$ and proceeding as in Lemma 4 of \cite{Slobodeanu2014} by evaluating $\partial_1 \U_1 -\partial_2 \U_2$, we
obtain a polynomial equation of degree 7 in $\U_3$, containing only basic coefficients and with leading term
$\O^7 \Q_3 (3p'+1)(3p'-1)^6(9p'-1)^6\U_3^7$.
Since we assume $\omega \theta \neq 0$, it follows that $\U_3$ is basic (hence we are done by Theorem \ref{th2}),
or that $\Q_3=0$. In the latter case we can choose a basic rotation making $\E_{12}=0$, under which equations (\ref{m_eq11}, \ref{m_eq9}) get simplified respectively to
\begin{eqnarray}\label{Th4_U1U2}
 \E_{13}\U_1-\E_{23}\U_2= \E_0 \U_3 + \B_{11}\frac{3p'+1}{3p'^2},\nonumber \\
 \E_{23}\U_1+\E_{13}\U_2= (\B_{10}-\B_9)\frac{3p'+1}{6p'^2}.
\end{eqnarray}
Unless the determinant of this system vanishes (in which case Lemma \ref{Theorem3} applies), we can solve (\ref{Th4_U1U2}) to obtain
expressions for $\U_1,\U_2$ which are linear in $\U_3$. Subsituting these in equations (\ref{m_eq1},\ref{m_eq2}) leads to two
quadratic equations in $\U_3$, with basic coefficients and with leading terms respectively
\begin{equation}
 \O \E_0 \frac{\E_{A3}}{\E_{13}^2+\E_{23}^2}\frac{(9p'-1)(3p'-1)}{3p'+1} \U_3^2, \ \ \ (A=1,2).
\end{equation}
It follows that $\U_3$ is basic (and we are done by Theorem \ref{th2}), unless $\E_0=0$, in which case equations (\ref{Th4_U1U2}) show that $\U_1$ and $\U_2$ are basic and we are done by Theorem \ref{th1}.
\end{proof}

\begin{te}\label{Theorem5}
If a shear-free perfect fluid obeys a linear equation of state $p=(\gamma-1)\mu + constant$ and if a Killing vector exists along the vorticity, then $\omega \theta=0$.
\end{te}

\begin{proof}
Assume that $\omega \theta \neq 0$. If a Killing vector along the vorticity exists, then $\dot{u}_3=0$ and we can impose all relations obtained in the proof of Theorem \ref{th2}:
\begin{equation}
\dot{u}_3= z_3 = q_3 = r_3= E_{13}=E_{23}=0,
\end{equation}
together with (\ref{th1_eq8bis}). Translating these in terms of basic variables, we obtain, besides $\dot{u}_3= z_3=0$
(i.e.~$\partial_3 \mu
=\partial_3 \theta=0$),
\begin{equation}
 \b_3=\Q_3=\R_3=\E_{13}=\E_{23}=0
\end{equation}
and two algebraic equations, namely (\ref{th1_eq8bis}) becoming (cf. also (42) in \cite{Slobodeanu2014})
\begin{eqnarray}\label{th5_extra1}
\fl \qquad (p'-1)(6p'+1)(\U_1^2+\U_2^2)-6(3p'+1)(\Q_1\U_1-\R_2 \U_2) \nonumber \\
\fl \qquad + 3(9p'-5) \frac{\lambda^8 \O^2}{(p+\mu)^2} + (3p'+1)(\theta^2-3\mu)\lambda^{-2} +6\E_3 -(3p'+1)\J = 0
\end{eqnarray}
and (\ref{basic_eq7}) simplifying to
\begin{equation}
 \B_6=\R_2\b_2-\Q_1\b_1.
\end{equation}
Under these restrictions the equations (\ref{m_eq11},\ref{m_eq9},\ref{m_eq10},\ref{m_eq7},\ref{m_eq8}) also tell us that
\begin{equation}
 \B_7=\B_8=\B_9=\B_{10}=\B_{11}=0.
\end{equation}
Furthermore, $\Q_3$ being 0, we have the freedom of performing an extra basic rotation in the $(1,2)$ plane, allowing us to
remove one of the basic isotropy-breaking variables, $\b_1, \b_2,\Q_1,\R_2,\partial_1 \J,\partial_2 \J,\E_{12}$ or $\E_0$. \\

We introduce now a new variable $\modU=\U_1^2+\U_2^2$, for which the time evolution can be written as
\begin{equation}\label{e0U_a}
 \partial_0 \modU=2 p'\theta \modU -\frac{3\lambda^5}{2(p+\mu)} (\b_1 \U_1+\b_2 \U_2).
\end{equation}
Our aim will be to construct a polynomial system with basic coefficients, in which the main variables are $\modU,\theta$ and $p+\mu$, while $\mu$ is given by (\ref{expl_eqstate}), namely
\begin{equation}\label{mu_expr}
 \mu = \frac{p+\mu}{p'+1}+\mu_0.
\end{equation}
In order to eliminate the variables $\U_1,\U_2$, we need the following equation (cf. also (43) in \cite{Slobodeanu2014}), obtained as linear combination of (\ref{m_eq3}) and (\ref{th5_extra1}),
\begin{eqnarray}
\fl \left( 4\,p'-1 \right)(\b_2\,  \U_1-\b_1\, \U_2)-\frac {42\,{p'}
^{3}+13\,{p'}^{2}-8\,p'+1} {3(3\,p'+1)}\O\modU
-\frac{\O}{3}(7\,p'-3)(\theta^2 -3 \mu)\lambda^{-2} \label{eq7bis}\\
\fl + 4\O(p+\mu)\lambda^{-2}
 - \frac{63\,{p'}^2+82\,p'-1 }{3\,p'+1} \frac{\lambda^{8}{\O}^{3}}{(p+\mu)^2} + \frac{21\,p'-1}{9}\O \, \J 
+ \frac{4(9\,p'-1)}{3(3\,p'+1)}\O \, \E_3 + \B_3=0. \nonumber
\end{eqnarray}
The subsequent time evolutions of this equation will be calculated using (\ref{utheta}--\ref{uU2}) and (\ref{e0mu},\ref{e0p}). 
The first element of this sequence, $\partial_0(\ref{eq7bis})$, is given by
(cf. also (44) in \cite{Slobodeanu2014})
\begin{eqnarray}
\fl \qquad (11 p'+ 1)(\b_1\,\U_1+\b_2\,\U_2)\frac{\lambda^5 \theta^{-1}}{p+\mu} 
+ \smfrac{4}{9} p' \left( 21 p'+11 \right) \modU \nonumber \\
\fl \qquad + \frac{8(3\,p'+1)}{9(3\,p' - 1)}\J
+ \frac {4\left( 9\,p'-1 \right)}{3\left(3\,p'-1 \right)}\E_3 
+ \frac {3\,p'+1}{3\,p'-1} \frac{\B_3}{\O}  \label{eq7bis0} \\
\fl \qquad + \smfrac{2}{3}\,\frac{\left( 21\,p'+1 \right)
 \left( 9\,p'-5 \right)  \left(p'+1 \right)}{p'\, \left(3\,p'-1 \right)} 
 \frac{ {\lambda}^{8}\O^{2}}{(p+\mu)^2} 
+ \smfrac{4}{3}\, \frac{\left(3\,p'+1 \right)\left(6\,p' + 1 \right)}{p'\left( 3\,p'-1 \right)}(p+\mu)\lambda^{-2}=0. \nonumber
\end{eqnarray}

First notice that, if $p'=-1/11$, $\partial_0$(\ref{eq7bis0}) and (\ref{e0U_a}) give rise to a new equation from which $\b_1 \U_1+\b_2 \U_2$ can
be calculated. The next derivative $\partial^2_0$(\ref{eq7bis0}) involves a term $\b_2 \U_1-\b_1 \U_2$, allowing one (as $\modU =
\U_1^2+\U_2^2\neq0$) to eliminate successively $\U_1,\U_2$, $\theta$ and $\modU$ from the sequence of derivatives $\partial_0^{(0)}$(\ref{eq7bis}), \ldots,
$\partial^{(4)}_0$(\ref{eq7bis}). Eventually, after substituting (\ref{expl_eqstate}), an equation in powers of $\lambda$ results,
with basic coefficients not all 0, showing that $\lambda$ is basic (details are available from the authors).

So henceforth we will assume $p'\neq -1/11$, allowing us to rewrite (\ref{e0U_a}) as
\begin{eqnarray}
\fl \partial_0 \modU =\frac{\theta}{(3p'-1)(11p'+1)} \left[\smfrac{4}{3}p'(3p'-1)(27p'+7) \modU + 2\frac{(3p'+1)(6p'+1)}{p'}(p+\mu)\lambda^{-2}\right.\nonumber \\
\fl \left. + \smfrac{4}{3}(3p'+1) \J +2(9p'-1)\E_3+\smfrac{3}{2}(3p'+1)\frac{\B_3}{\O}+\frac{(21p'+1)(9p'-5)(p'+1)}{p'}
\frac{\lambda^8 \O^2}{(p+\mu)^2}\right]. \label{e0U_b}
\end{eqnarray}

At this stage it becomes advantageous to apply a basic rotation such that, for example $\b_2=\b_1$.

In the following lines we only describe the outline
of the proof, as the output of the calculations is far too lengthy for publication. A Maple or Mathematica worksheet with all the details can be obtained from the authors.
First we should consider the special case
$p'=1/4$: the linear terms in $\U_1,\U_2$ are then absent from (\ref{eq7bis}), but reappear in
its evolution via (\ref{e0U_a}) and $\partial^{(2)}_0$(\ref{eq7bis0}). Similar to the case $p' = 1/11$ elimination of
$\U_1,\U_2$, $\theta$ and $\modU$ from the sequence of derivatives $\partial_0^{(0)}$(\ref{eq7bis}), \ldots ,
$\partial^{(4)}_0$(\ref{eq7bis}) leads then to $\lambda$ being basic, whence $\theta=0$.\\

When $p'\neq-1/11$ and $p'\neq1/4$, we first proceed as in \cite{Slobodeanu2014}: using (\ref{eq7bis},\ref{eq7bis0}) to eliminate the linear $\U_1,\U_2$ terms from
the sequence $\partial_0^{(2)}$(\ref{eq7bis}), \ldots, $\partial^{(4)}_0$(\ref{eq7bis}), we obtain three relations of the form
$\mathcal{P}_i=a_i \modU^2+b_i\modU \theta^2 +c_i \modU +d_i \theta^2 +e_i=0$ ($i=1,2, 3$), with $a_i,b_i,c_i,d_i,e_i$ polynomials having basic coefficients of degree 2 and 20 in respectively
$p+\mu$ and $\lambda$. Eliminating $\theta$ results
in two relations $\mathcal{R}_i(\modU,p+\mu, \lambda)=0$ ($i=1, 2$), with $\mathcal{R}_1, \mathcal{R}_2$ polynomials of third degree in
$\modU$ and having degrees $9$ and $30$ in respectively $p+\mu$ and $\lambda$. Their resultant $\mathcal{F}(p+\mu, \lambda)$ with respect to $\modU$ factorises as follows over $\mathbb{Q}$:
\begin{eqnarray}
\fl \mathcal{F} = \O^9 p'^{13}(p+\mu)^{18}\lambda^{18}(4p'-1)(6p'+1)(3p'+1)^2(3p'-1)^9\nonumber \\
\times (21p'+11)^4(11p'+1)^6(117 p'^2+69p'+2)
(96 p'^2+47 p'+1) \mathcal{F}_1 \mathcal{F}_2
\end{eqnarray}
with $\mathcal{F}_1$, $\mathcal{F}_2$ respectively of degrees (6,20) and (21,70) in $p+\mu$ and $\lambda$ and both having basic coefficients
depending on $\O, \J, \b_1, \E_3, \B_3$. Using (\ref{expl_eqstate}) any such polynomial in $p+\mu$ and $\lambda$ will be written as $\sum_{i,j} c_{i,j} \lambda^{i r+j}$ with $c_{i,j}$ basic functions.
The remaining part of the proof is based on Lemma 3 of \cite{Slobodeanu2014}, which essentially says that a finite sum $\sum_{i,j} c_{i,j} \lambda^{i r+j}$ of products of basic functions and real powers
of a (non-basic) function $\lambda$ can only be 0 if all coefficients vanish: if a `reference coefficient' $c_{i_0,j_0} \neq 0$ and if for all $(i,j) \neq (i_0,j_0)$ there are no cancellations corresponding to $i r + j = i_0 r + j_0$, then
$\lambda$ is basic. As cancellations can only occur for rational values of $r$, this implies a.o.~that for irrational $r$ all $c_{i,j}$ must be identically 0.\\

While the cases $p'\in\{-\smfrac{1}{3},-\smfrac{1}{11},0,\smfrac{1}{4},\smfrac{1}{3}\}$ have been dealt with before,
the special cases $6p'+1=0$, $117{p'}^2+69p'+2=0$, $96{p'}^2+47p'+1=0$ and $21p'+11=0$ correspond to the situation where
the degree w.r.t.~$\modU$ of $\mathcal{R}_1, \mathcal{R}_2$ decreases to 2 or 1 (for $21p'+11=0$). Calculating the resultant of
$\mathcal{R}_1, \mathcal{R}_2$, after simplifying first w.r.t.~the given $p'$ relations, results in $\lambda$ being a root of
a polynomial in some fractional power of $\lambda$ (and with basic coefficients not all being 0). It follows that $\lambda$ is basic, whence
$\theta=0$.\\

The case $\mathcal{F}_1=0$ is slightly more complicated: after substituting $p$ and $\mu$ as functions of $\lambda$ via (\ref{expl_eqstate}),
the occurring terms belong to the set $\{\lambda^{6r},\lambda^{5r},\lambda^{5r+2},\lambda^{4r+4},\lambda^{4r+2},
\lambda^{3r+10},\lambda^{2r+10},\lambda^{2r+12},\lambda^{20}\}$, with the coefficients $c_{0,20}$ and $c_{6,0}$ of $\lambda^{20}$ and
$\lambda^{6r}$ polynomials in $r$ having no common factor and the former being irreducible over $\mathbb{Q}$. The case $c_{0,20}=0$ hereby
being excluded, the case
$c_{0,20}\neq 0$ implies that the $\lambda^{20}$ term must cancel with one of the remaining terms in $\mathcal{F}_1$, leading to $r\in
\{\smfrac{9}{2},\smfrac{10}{3},\smfrac{18}{5},4,5\}$. While the cases $r\in\{\smfrac{10}{3},4\}$ ($p'\in\{\smfrac{1}{9},\smfrac{1}{3}\}$) have
been excluded before, the cases $r\in \{\smfrac{9}{2},\smfrac{18}{5},5\}$ can easily be excluded by direct substitution in $\mathcal{F}_1$ and
by verifying that the resulting
polynomial in (some rational power of) $\lambda$ is not identically 0.\\

The hardest case $\mathcal{F}_2=0$ can be dealt with in a similar way. First notice that the coefficients $c_{0,70}, c_{3,60}$ and $c_{21,0}$ of
$\mathcal{F}_2=\sum_{i,j} c_{i,j} \lambda^{i r+j}$ are polynomials in $r$ of degrees respectively 42, 45 and 57, with the rational roots 
belonging either to the set
$\{0,2,\smfrac{5}{2},\smfrac{8}{3},\smfrac{30}{11},3,\smfrac{10}{3},\smfrac{15}{4},4\}$
of already excluded $r$-values, or to the set $\{\smfrac{8}{3},\smfrac{20}{7},\smfrac{14}{3}\}$. Again by direct substitution in $\mathcal{F}_2$ 
it is easy to show that the latter three values of $r$ are excluded, while a simple
evaluation of resultants shows that $c_{0,70}, c_{3,60}$ and $c_{21,0}$ have no common irrational roots (besides those corresponding to the previously excluded case $117{p'}^2+69p'+2=0$).
It follows that each of the terms $\lambda^{70}, \lambda^{60+3 r}$ and
$\lambda^{21r}$ must cancel with one of the other terms in $\mathcal{F}_2$, yielding three large sets of $r$-values to be investigated.
However the intersection of the three sets only contains the excluded value $r=\smfrac{10}{3}$ and therefore
$\mathcal{F}_2=0$ implies that $\lambda$ is basic, whence $\theta=0$.
\end{proof}

In Section 6 of [30] a very similar proof to this final `elusive case' was given with the assumption that the cosmological constant vanishes and which does not cover the exceptional cases
$p' \in \{ -\frac{1}{6},-\frac{1}{11},-\frac{1}{21}, \frac{1}{4} \}$.
We notice that the system obtained by iterated propagation of (73) was there seen as a system in $\theta^2$ and $\partial_0 \theta$. The different choice of variables employed here allowed a unitary treatment of the
cases $p' \in \{ -\frac{1}{11},  \frac{1}{4} \}$, while $p' = -\frac{1}{6}$ has been easily eliminated and $p' = -\frac{1}{21}$ no longer occurs.

\begin{re}
One could wonder whether it is always possible to fix the tetrad -- as we did in the proof of Theorem \ref{Theorem5} -- such that all basic variables become invariants and hence such that, in the case of a Killing vector along the vorticity,
all the occurring $\partial_3$ derivatives become
0. It is easy to see, \emph{even for a non-linear equation of state}, that the exceptional situation, in which all the basic isotropy-breaking variables,
$\b_1=\b_2=\Q_1=\R_2=\E_{12}=\E_0=\partial_1 \J=\partial_2 \J$ vanish, is inconsistent:
(\ref{basic_eq27},\ref{basic_eq28}) imply then $\B_{12}=\B_{13}=0$, turning (\ref{m_eq12},\ref{m_eq13}) into a homogeneous system in $\U_1,\U_2$,
the determinant of which is positive definite (and hence the acceleration is parallel to the vorticity), unless $9 G p'-9 {p'}^2+1 = 0$ and
\begin{equation*}
 (p'+1)(81{p'}^2-5)\O^2\lambda^8-6{p'}^2(p+\mu)^2 \E_3 = 0.
\end{equation*}
Propagating this second equation along $\w{u}$ and simplifying the result by means of $9 G p'-9 {p'}^2+1 = 0$ leads then to a contradiction.

\end{re}

\section{Conclusion and discussion}

For shear-free perfect fluids obeying a barotropic equation of state (with $p+\mu\neq 0$) and obeying the Einstein field equations (with or without 
cosmological constant) we first have demonstrated two theorems, showing that
(Theorem \ref{th1}) $\omega \theta=0$ once $\dot{u}_1 / (\lambda p')$ and  $\dot{u}_2 / (\lambda p')$ are basic and (Theorem \ref{th2}) that either 
$\omega \theta=0$ or a Killing vector along the vorticity vector exists once $\dot{u}_3 / (\lambda p')$ is basic. In particular, Theorem \ref{th2} shows that (when $\omega \theta \neq 0$) the existence of a Killing vector along the vorticity is equivalent to the orthogonality of acceleration and vorticity. 
Next we have demonstrated (Theorem 3 and 4) that $\omega \theta=0$ once the equation of state is linear:
$p=(\gamma -1) \mu + p_0$, covering in the new proof all the exceptional cases of \cite{Slobodeanu2014} and generalising the result to the 
possible presence of a cosmological constant (absorbed in $p_0$). While doing so we
generalised the formalism of \cite{Slobodeanu2014} to general equations of state, hoping herewith (and with the aid of theorems \ref{th1} and \ref{th2}) to 
have provided the interested reader with a new technique to tackle the Shear-free Fluid Conjecture in its full generality. 

In section 5 we have demonstrated that the assumption of a linear equation of state together with $\omega \theta \neq 0$ implies $\Q_3=0$.
Lemma \ref{Theorem3} was then used to reduce the problem to the situation where a Killing vector exists along the vorticity. We are convinced that this lemma is also valid for a general barotropic equation of state (although we have not been able
to provide a detailed proof of this claim), and hence may play a key role in the general proof. The hardest part will then undoubtedly remain 
to prove the conjecture in the case where there is a Killing vector along the vorticity ...\\

The interested reader can obtain Maple or Mathematica worksheets with full details of all the proofs from the authors.

\section{Appendix 1}

The following is the initial set of equations describing a shear-free perfect fluid with a general barotropic equation of state $p=p(\mu)$, assuming throughout $\omega, p'$ and $3 p' +1 \neq 0$.\\

\textbf{a)} \  evolution equations ($A=1,2$, $\alpha=1,2,3$)
\begin{eqnarray}
\partial_{0}\mu &=& -(p+\mu) \theta, \quad \textrm{(conservation of mass)} \label{e0mu} \\
\partial_{0}p &=& -(p+\mu) p' \theta, \label{e0p} \\
\partial_{0}\theta&=&-\tfrac{1}{3}\theta^{2} + 2\omega^{2}-\tfrac{1}{2}(\mu+3p) + j, \
\textrm{(Raychaudhuri eq.)}  \label{Raych} \\
\partial_0 \dot{u}_{\alpha} &=& p' z_{\alpha} - G\theta \dot{u}_{\alpha}, \label{e0acc} \\
\partial_0 \omega &=&\tfrac{1}{3} \omega \theta(3p'-2), \label{e0om} \\
\partial_0 r_{\alpha}   &=& - \frac{1}{3} z_{\alpha} - \frac{\theta}{3} (\dot{u}_{\alpha} + r_{\alpha}), \label{e0r}\\
\partial_0 q_{\alpha} &=& \frac{1}{3} z_{\alpha} + \frac{\theta}{3} (\dot{u}_{\alpha}-q_{\alpha}), \label{e0q} \\
\partial_0 n &=& -\frac{\theta}{3} n, \label{e0n} \\
\partial_0 z_1 &=& \theta (p'-1) z_1-\tfrac{1}{2}\omega(9p' -1) z_2+ \tfrac{1}{2}\theta\omega(9G-2)\dot{u}_2, \label{e0z1}\\
\partial_0 z_2 &=&  \theta (p'-1) z_2+\tfrac{1}{2}\omega(9p' -1) z_1- \tfrac{1}{2}\theta\omega(9G-2)\dot{u}_1, \label{e0z2} \\
\partial_0 z_3 &=& \theta (p' -1) z_3, \label{e0z3} \\
\partial_0 j &=& \theta \left[G_p(p+\mu)- 2G + 1 \right]\dot{\w{u}}^2 -(2G-1)\dot{\w{u}}\cdot \w{z} \nonumber \\
&-& \theta\left[\tfrac{1}{3}(3G+1)j - p'(9p'-1)\omega^2\right], \label{e0jfirst}
\end{eqnarray}

$\diamondsuit$ the $\dot{\w{E}}$ second Bianchi identities,
\begin{eqnarray}
\partial_0 E_{AA} &=& \frac{1}{3(3p'+1)}[(18p'^2-3p'+G-1)\theta\omega^2-
(3G+9p'+1)\theta E_{AA} \nonumber\\
 & &  -2G(3\dot{u}_A  z_A -\dot{\w{u}} \cdot \w{z}) -(2G-(p+\mu)G_p)\theta
(3\dot{u}_A^{~2}-\dot{\w{u}}^2) ], \\
\partial_0 E_{33} &=& -\frac{1}{3(3p'+1)}[2(18p'^2-3p'+G-1)\theta\omega^2+
(3G+9p'+1)\theta E_{33} \nonumber\\
 & &  +2G(3\dot{u}_3  z_3 -\dot{\w{u}} \cdot \w{z}) +(2G-(p+\mu)G_p)\theta
(3\dot{u}_3^{~2}-\dot{\w{u}}^2) ],\label{e0E33} \\
\partial_0 E_{\alpha\beta} &=& -\frac{1}{3p'+1} [(G+3p'+\tfrac{1}{3})\theta
E_{\alpha\beta}+G (\dot{u}_\alpha z_\beta+\dot{u}_\beta z_\alpha) \nonumber\\
 & &  + (2G-(p+\mu)G_p) \theta\dot{u}_\alpha \dot{u}_\beta ],
\label{e0Ealfabeta}
\end{eqnarray}

\textbf{b)} \ spatial equations
\begin{eqnarray}
\partial_\alpha  p = - (p+\mu) \dot{u}_{\alpha}, \quad
 \textrm{(Euler equations)} \label{gradp} \\
\partial_\alpha \mu = -\frac{p+\mu}{p'} \dot{u}_{\alpha},  \quad \label{gradmu} \\
\partial_\alpha \theta = z_\alpha, \quad
(\textrm{definition of } \w{z}) \label{gradtheta}
\end{eqnarray}
\begin{eqnarray}
\partial_1 \omega = \tfrac{2}{3} z_2-\omega(q_1+2 \dot{u}_1), \quad & \big(\textrm{\small{(02)--Einstein field equation}}\big) \label{gradom1}\\
\partial_2 \omega = - \tfrac{2}{3} z_1+\omega (r_2-2\dot{u}_2),
\quad & \big(\textrm{\small{(01)--Einstein field equation}}\big)\label{gradom2}\\
\partial_3 \omega = \omega( \dot{u}_3+r_3-q_3), & \label{gradom3}
\end{eqnarray}
\begin{eqnarray}
\partial_{1}\dot{u}_{1} = \tfrac{1}{3}(j - \omega^2) - q_{2}\dot{u}_{2} + r_{3}\dot{u}_{3} -\dot{u}_1^2+E_{11}, & \big(\textrm{\small{(11)--Einstein field eq.}}\big) \label{e1u1}\\
\partial_{2}\dot{u}_{2} = \tfrac{1}{3}(j - \omega^2) - q_3\dot{u}_3 + r_1\dot{u}_1 - \dot{u}_2^2 + E_{22}, & \big(\textrm{\small{(22)--Einstein field eq.}}\big)\label{e2u2}\\
\partial_{3}\dot{u}_{3} = \tfrac{1}{3}(j + 2\omega^2) - q_1\dot{u}_1 + r_2 \dot{u}_2  - \dot{u}_3^2 + E_{33}, & \big(\textrm{\small{(33)--Einstein field eq.}}\big)\label{u33}\\
\partial_1 \dot{u}_2 = -p'\omega \theta + q_2 \dot{u}_1+\tfrac{1}{2} n_{33} \dot{u}_3-\dot{u}_1 \dot{u}_2+E_{12}, & \big(\textrm{\small{(12)--Einstein field eq.}}\big)\label{e1u2}\\
\partial_2 \dot{u}_1 = p'\omega \theta-r_1 \dot{u}_2-\tfrac{1}{2} n_{33} \dot{u}_3-\dot{u}_1 \dot{u}_2+E_{12}, & \big(\textrm{\small{(21)--Einstein field eq.}}\big)\label{e2u1}\\
\partial_1 \dot{u}_3 = -\tfrac{1}{2} n_{33} \dot{u}_2 - r_3 \dot{u}_1-\dot{u}_1 \dot{u}_3+E_{13},  & \big(\textrm{\small{(13)--Einstein field eq.}}\big) \label{e1u3} \\
\partial_2 \dot{u}_3 = \tfrac{1}{2}n_{33}\dot{u}_1 + q_3 \dot{u}_2-\dot{u}_2 \dot{u}_3+E_{23},  & \big(\textrm{\small{(23)--Einstein field eq.}}\big) \label{e2u3} \\
\partial_3 \dot{u}_1 = -(\tfrac{1}{2} n_{33} - n)\dot{u}_2 + q_1 \dot{u}_3 - \dot{u}_1 \dot{u}_3+E_{13}, & \big(\textrm{\small{(31)--Einstein field eq.}}\big) \label{e3u1} \\
\partial_3 \dot{u}_2 = (\tfrac{1}{2}n_{33} - n) \dot{u}_1 - r_2 \dot{u}_3 - \dot{u}_2 \dot{u}_3 + E_{23}, & \big(\textrm{\small{(32)--Einstein field eq.}}\big) \label{e3u2}
\end{eqnarray}
\begin{eqnarray}
\partial_1 j &=& p'\theta z_1-\frac{1}{6} \omega (27p'+13)z_2+\frac{1}{3} (18 \omega^2+\theta^2-3j-3\mu)\dot{u}_1 - \frac{p+\mu}{2p'}\dot{u}_1 \nonumber \\
&+&\frac{1}{2}\theta \omega(9G-2) \dot{u}_2 +4 \omega^2 q_1, \label{e1j}\\
\partial_2 j &=& p'\theta z_2+\frac{1}{6} \omega (27p'+13)z_1 + \frac{1}{3}(18 \omega^2+\theta^2-3j-3\mu)\dot{u}_2 - \frac{p+\mu}{2p'}\dot{u}_2\nonumber \\
&-&\frac{1}{2}\theta \omega(9G-2) \dot{u}_1 -4 \omega^2 r_2, \label{e2j}\\
\partial_3 j &=& p'\theta z_3+\frac{1}{3}(\theta^2-18 \omega^2-3j-3\mu)\dot{u}_3 -\frac{p+\mu}{2p'} \dot{u}_3-4 (r_3-q_3) \omega^2, \label{e3j}
\end{eqnarray}
\begin{eqnarray}
\partial_1 z_1  &=&
\frac{1}{1+3p'}[\theta(3G-2)E_{11}+(2G-(p+\mu)G_p)\theta(3\dot{u}_1^{~2}-\dot{\w{u}}^{2})\nonumber\\
& &
+2(3G-3p'-1)\dot{u}_1z_1-2 G\dot{\w{u}}\cdot \w{z} +\theta\omega^2 (9p'^2+6p'-G-1)] \nonumber \label{e1Z1} \\
& &  + r_3z_3 - q_2z_2 \label{Z11}, \\
\partial_2 z_2 &=&
\frac{1}{1+3p'}[\theta(3G-2)E_{22}+(2G-(p+\mu)G_p)\theta(3\dot{u}_2^{~2}-\dot{\w{u}}^2)\nonumber \\
& &
+2(3G-3p'-1)\dot{u}_2z_2-2 G\dot{\w{u}}\cdot \w{z}+\theta\omega^2(9p'^2+6p'-G-1)]\nonumber \label{e2Z2}\\
& & + r_1z_1-q_3z_3\label{Z22},\\
\partial_3 z_3
&=&\frac{1}{1+3p'}[\theta(3G-2)E_{33}+(2G-(p+\mu)G_p)\theta(3\dot{u}_3^{~2}-\dot{\w{u}}^{2})\nonumber \\
& &
+2(3G-3p'-1)\dot{u}_3z_3-2 G\dot{\w{u}}\cdot \w{z}+\theta\omega^2(9p'^2- 6p' + 2G + 1)]\nonumber \\
& & + r_2z_2 - q_1 z_1, \label{Z33}
\end{eqnarray}
\begin{eqnarray}
 \partial_1 z_2 &=& q_2z_1+\frac{n_{33}}{2}z_3+\frac{\omega}{6}(2\theta^2-12 \omega^2-6 j + 9 p+3 \mu) \nonumber \\
 &+& \frac{1}{1+3p'}[(3G-1-3p')(\dot{u}_2 z_1+ \dot{u}_1
z_2)+3(2G-(p+\mu)G_p)\theta\dot{u}_1 \dot{u}_2\nonumber \\
 & & +\theta (3G-2)E_{12}]\label{e1Z2},\\
 \partial_2 z_1 &=& -r_1 z_2-\frac{n_{33}}{2}-\frac{\omega}{6}(2\theta^2-12 \omega^2-6 j + 9p+3 \mu) \nonumber \\
& & + \frac{1}{1+3p'}[(3 G-1-3p')(\dot{u}_2 z_1+ \dot{u}_1
z_2)+3p'(2G-(p+\mu)G_p)\theta\dot{u}_1 \dot{u}_2\nonumber \\
& & +\theta (3G-2)E_{12}]\label{e2Z1},\\
 \partial_3 z_1 &=& q_1 z_3+\left(n-\frac{n_{33}}{2}\right)z_2
 +\frac{1}{1+3p'}[(3G-1-3p')(\dot{u}_3 z_1+ \dot{u}_1z_3) \nonumber \\
& & + 3\theta(2G-(p+\mu)G_p)\dot{u}_1\dot{u}_3+\theta(3G-2) E_{13}],\\
 \partial_3 z_2 &=& -r_2 z_3-\left(n-\frac{n_{33}}{2}\right)z_1
+\frac{1}{1+3p'}[(3G-1-3p')(\dot{u}_3 z_2+ \dot{u}_2z_3)\nonumber \\
& & + 3\theta(2G-(p+\mu)G_p)\dot{u}_2\dot{u}_3+\theta(3G-2) E_{23}] ,\\
 \partial_1 z_3 &=& -r_3 z_1-\frac{n_{33}}{2}z_2+\frac{1}{1+3p'}[(3G-1-3p')(\dot{u}_3 z_1+ \dot{u}_1z_3)\nonumber \\
& & + 3\theta(2G-(p+\mu)G_p)\dot{u}_1\dot{u}_3+\theta(3G-2)
E_{13}],\label{e1Z3} \\
 \partial_2 z_3 &=& q_3 z_2+\frac{n_{33}}{2}z_1 +\frac{1}{1+3p'}[(3G-1-3p')(\dot{u}_3 z_2+ \dot{u}_2z_3) \nonumber \\
& & + 3\theta(2G-(p+\mu)G_p)\dot{u}_2\dot{u}_3+\theta(3G-2)E_{23}]\label{e2Z3},
\end{eqnarray}

\medskip 

$\diamondsuit$ \ two equations obtained as linear combinations of the (12)--Einstein field equation and one of the Jacobi equations,
\begin{eqnarray}\label{e2q1}
\partial_2q_1+\smfrac{1}{2}\partial_3 n_{33}-r_2(r_1+q_1)-n(q_3+r_3)+n_{33}q_3-\smfrac{1}{3}\omega \theta +E_{12} =0,
\end{eqnarray}
\begin{eqnarray}\label{e1r2}
\partial_1r_2+\smfrac{1}{2}\partial_3 n_{33}+q_1(r_2+q_2)+n(q_3+r_3)-n_{33}r_3-\smfrac{1}{3}\omega \theta -E_{12} =0,
\end{eqnarray}

$\diamondsuit$ \ linear combinations of the (13)-- and (23)--Einstein field equations and the Jacobi equations,
\begin{eqnarray}
\partial_2 r_3 = -\smfrac{1}{2}\partial_1 n_{33}-q_1 n_{33}-q_2(q_3+r_3)+E_{23}, \label{e2r3} \\
%
\partial_1 q_3 = -\smfrac{1}{2}\partial_2 n_{33}+r_2 n_{33}+r_1(q_3+r_3)-E_{13}, \label{e1q3}
\end{eqnarray}

\textbf{c)} \ the (03)--Einstein field equation
\begin{equation}\label{defn33}
n_{33} = \frac{2}{3\omega} z_3, \quad
\end{equation}

\textbf{d)} \ remaining combinations of the Jacobi equations and the $(1,3)$, $(2,3)$ and $(\alpha,\alpha)$--Einstein field equations
\begin{eqnarray}
 \partial_1 n+\partial_3 q_2 = \smfrac{1}{2} \partial_1 n_{33} +n(r_1-q_1) +r_3(r_2+q_2)+q_1 n_{33}-E_{23} \label{cons1}, \\
 \partial_2 n+\partial_3 r_1 = \smfrac{1}{2} \partial_2 n_{33} +n(r_2-q_2) -q_3(r_1+q_1)-r_2 n_{33}+E_{13} \label{cons2}
\end{eqnarray}
and
\begin{eqnarray}
 \partial_3 q_3-\partial_2 r_2 &=& E_{11}-\frac{\mu}{3}+\frac{\theta^2}{9}+\frac{n_{33}^2}{4}-r_2^2-q_3^2 +r_1q_1\label{ein11bis}, \\
 \partial_1 q_1-\partial_3 r_3 &=& E_{22}-\frac{\mu}{3}+\frac{\theta^2}{9}+\frac{n_{33}^2}{4}-r_3^2-q_1^2 +r_2q_2\label{ein22bis}, \\
 \partial_2 q_2-\partial_1 r_1 &=& E_{33}-\frac{\mu}{3}+\frac{\theta^2}{9}-\frac{3 n_{33}^2}{4} + n \, n_{33}+3\omega^2-
 q_2^2-r_1^2 +r_3 q_3, \label{ein33bis}
\end{eqnarray}

\textbf{e)} the `$\dot{\w{H}}$' second Bianchi identities
\begin{eqnarray}
 \fl \partial_0 H_{11} + \partial_2 E_{13} - \partial_3 E_{12}&=&  E_{11} n -\left(n-\frac{n_{33}}{2}\right) E_{22} - \frac{1}{2} E_{33} n_{33}+(q_3+2\u_3)E_{12}\nonumber \\
 & & +(r_2-2\u_2) E_{13} -(r_1+q_1)E_{23}-\theta H_{11}-\omega H_{12}, \label{H11_0}\\
 \fl \partial_0 H_{22} + \partial_3 E_{12} - \partial_1 E_{23}&=&  E_{22} n - \left(n-\frac{n_{33}}{2}\right) E_{11}-\frac{1}{2} E_{33} n_{33}+(r_3-2\u_3)E_{12}\nonumber \\
 & & +(q_1+2\u_1) E_{23} -(r_2 + q_2)E_{13}-\theta H_{22}+\omega H_{12}, \label{H22_0}\\
 \fl \partial_0  H_{12}-\partial_3 E_{22}+\partial_2 E_{23} &=& (q_3+2\u_3)E_{22}-(q_3-\u_3)E_{33}+\left(2 n-\frac{n_{33}}{2}\right)E_{12}-\frac{p+\mu}{6 p^\prime} \u_3\nonumber \\
 & & +(r_1+\u_1)E_{13}+(2r_2-\u_2)E_{23}+(H_{11}-H_{33})\omega - H_{12} \theta, \label{H12_0}\\
 \fl \partial_0 H_{13} +\partial_2 E_{33}-\partial_3 E_{23} &=& (\u_2-2r_2)E_{22}-(r_2-2\u_2)E_{11}+(q_1-\u_1) E_{12}
 +\frac{p+\mu}{6p^\prime} \u_2\nonumber \\
 & & +(2 q_3+\u_3) E_{23}+\left(n+\frac{n_{33}}{2}\right) E_{13}-H_{13}\theta+H_{23} \omega, \label{H13_0}\\
  \fl \partial_0 H_{23} -\partial_1 E_{33}+\partial_3 E_{13} &=& -(\u_1+2q_1)E_{11}-(q_1+2\u_1)E_{22}+(r_2+\u_2) E_{12}-\frac{p+\mu}{6 p^\prime}\u_1\nonumber \\
 & & +(2 r_3-\u_3) E_{13}+\left(n+\frac{n_{33}}{2}\right) E_{23}-H_{23}\theta-H_{13} \omega, \label{H23_0}
\end{eqnarray}

\textbf{f)} the `$\w{\nabla \cdot H}$' second Bianchi equations (with $\w{\nabla \cdot H}_3$ becoming an identity under these two)
\begin{eqnarray}
\fl \partial_3 q_1+\partial_1 r_3 = -4 E_{13}+\frac{3G-2}{3(3p'+1)} \frac{\theta}{\omega} E_{23}+3(\u_1\u_3 +\u_1 r_3-\u_3 q_1)
+\frac{2G-G_p(p+\mu)}{3p'+1}\frac{\theta}{\omega}\u_2\u_3 \nonumber \\
+\frac{G z_3\u_2}{\omega(3p'+1)}+\frac{(G-3p'-1) z_2\u_3}{\omega(3p'+1)}-\frac{z_1 z_3}{9\omega^2}
+\frac{(q_3-4r_3)z_2}{3\omega}+\frac{r_2 z_3}{3\omega} \nonumber \\
+r_1 r_3+q_1 q_3+q_3 r_1-n r_2 \label{divH1},\\
\fl \partial_3 r_2+\partial_2 q_3 = 4 E_{23}+\frac{3G-2}{3(3p'+1)} \frac{\theta}{\omega} E_{13}-3(\u_2\u_3 -\u_2 q_3+\u_3 r_2)
+\frac{2G-G_p(p+\mu)}{3p'+1}\frac{\theta}{\omega}\u_1\u_3 \nonumber \\
+\frac{G z_3\u_1}{\omega(3p'+1)}+\frac{(G-3p'-1) z_1\u_3}{\omega(3p'+1)}+\frac{z_2 z_3}{9\omega^2}
-\frac{(r_3-4q_3)z_1}{3\omega}-\frac{q_1 z_3}{3\omega} \nonumber \\
-r_2 r_3-q_2 q_3-r_3 q_2+n q_1, \label{divH2}
\end{eqnarray}

$\diamondsuit$ \  `$\w{\nabla \cdot E}$' second Bianchi equations (taking into account (\ref{def_H}))
\begin{eqnarray}
\fl \partial_{\beta} {E^\beta}_1 +E_{11}(2 q_1 -r_1)+E_{12}(2 q_2-r_2)+E_{13}(q_3-2r_3)+E_{22}(r_1+q_1)+E_{23}(n_{33}-n) \nonumber \\
+\omega z_2-3 \omega^2 q_1+\frac{\mu+p}{3 p'} \u_1 =0, \label{divE1} \\
\fl \partial_{\beta} {E^\beta}_2 -E_{22}(2 r_2 -q_2)-E_{12}(2 r_1-q_1)-E_{23}(r_3-2q_3)-E_{11}(r_2+q_2)-E_{13}(n_{33}-n) \nonumber \\
-\omega z_1+3 \omega^2 r_2+\frac{\mu+p}{3 p'} \u_2 =0, \label{divE2} \\
\fl \partial_{\beta} {E^\beta}_3 +E_{13}(2 q_1 -r_1) - E_{23}(2 r_2 -q_2)+E_{33}(2q_3-r_3)+E_{11}(r_3+q_3) \nonumber \\
-3 \omega^2(q_3-r_3-2\u_3)+\frac{\mu+p}{3 p'} \u_3 =0. \label{divE3}
\end{eqnarray}

\section{Appendix 2}

Here we present the purely basic differential equations accompanying the algebraic
relations constructed in section 3.

The first set contains the definitions of the basic variables $\B_1, \ldots, \B_{16}$:
\begin{eqnarray}
\fl \X(\b_3) - \Z(\b_1)=(\R_3-\Q_3)\b_1-\N\b_2-\Q_1\b_3-\B_1, \label{basic_eq11} \\
\fl \Y(\b_3) - \Z(\b_2)=(\R_3+\Q_3)\b_2+\N\b_1+\R_2\b_3-\B_2, \label{basic_eq12} \\
\fl \X(\b_2) - \Y(\b_1) = \R_1\b_2+\Q_2\b_1-\frac{\b_3^2}{2\O}+\B_3, \label{basic_eq13} \\
\fl \X(\b_1) - \Y(\b_2)= -\R_1\b_1-\Q_2\b_2+2 \Q_3\b_3-\B_4, \label{basic_eq14} \\
\fl \X(\b_2) + \Y(\b_1)= \Q_2\b_1-\R_1\b_2+\B_5, \label{basic_eq15} \\
\fl \X(\b_1) + \Y(\b_2) = \R_1\b_1-\Q_2\b_2-2 \R_3\b_3+\B_6, \label{basic_eq16} \\
\fl \X(\b_3)= (\R_3-\Q_3)\b_1+\frac{\b_2\b_3}{4\O}+\B_7, \label{basic_eq17} \\
\fl \Y(\b_3)= (\R_3+\Q_3)\b_2-\frac{\b_1\b_3}{4\O}+\B_8, \label{basic_eq18} \\
\fl \Y(\E_{13}) - \Z(\E_{12}) = \left(\N+\frac{\b_3}{8 \O}\right)\E_0+\frac{3 \b_3}{8 \O} \E_3+(\Q_3+\R_3)\E_{12} \nonumber \\
 +\R_2\E_{13}-(\Q_1+\R_1)\E_{23}+\B_9, \label{basic_eq24} \\
 \fl \X(\E_{23}) - \Z(\E_{12}) = \left(\N+\frac{\b_3}{8 \O}\right)\E_0-\frac{3 \b_3}{8 \O} \E_3-(\Q_3-\R_3)\E_{12} \nonumber \\
 -\Q_1\E_{23}+(\Q_2+\R_2)\E_{13}-\B_{10}, \label{basic_eq25} \\
 %
%
\fl \Y(\E_3-\tfrac{1}{6}\J) - \Z(\E_{23}) =\half \R_2(\E_0+3\E_3) +\Q_1 \E_{12} \nonumber \\
+\left(\N-\frac{\b_3}{4\O}\right)\E_{13}+2 (\Q_3+\R_3) \E_{23}  + \B_{12},\label{basic_eq27} \\
\fl \X(\E_3-\tfrac{1}{6}\J) + \Z(\E_{13}) =-\half \Q_1(\E_0-3\E_3) +\R_2 \E_{12}\nonumber \\
+\left(\N-\frac{\b_3}{4\O}\right)\E_{23}+2 (\Q_3-\R_3) \E_{13} - \B_{13},\label{basic_eq28} \\
\fl \half \X(\E_3+\E_0) + \Y(\E_{12}) = \R_1\E_0-2\Q_2\E_{12}-(\Q_3+\R_3) \E_{13}+\frac{3\b_3}{4\O} \E_{23}- \B_{13}-\frac{\B_{14}}{6}, \label{basic_eq29} \\
\fl \half \Y(\E_3-\E_0) + \X(\E_{12}) = \Q_2\E_0+2\R_1\E_{12}+(\Q_3-\R_3) \E_{23}-\frac{3\b_3}{4\O} \E_{13}+\B_{12}-\frac{\B_{15}}{6}, \label{basic_eq30} \\
\fl \Z(\E_0-\E_3)-2 \X(\E_{13}) = (\Q_3-\R_3)(\E_0-3 \E_3)+\left(4\N+\frac{\b_3}{2\O}\right)\E_{12}+4\Q_1 \E_{13}+2\Q_2 \E_{23} \nonumber \\ 
+2 \B_{11}+\tfrac{1}{3}\B_{16}, \label{basic_eq31} \\
\fl \Z(\E_0+\E_3)+2 \Y(\E_{23}) = -(\Q_3+\R_3)(\E_0+3 \E_3)+\left(4\N+\frac{\b_3}{2\O}\right)\E_{12}+4\R_2 \E_{23}+2\R_1 \E_{13}\nonumber \\
+2 \B_{11} - \tfrac{1}{3}\B_{16}, \label{basic_eq26}
\end{eqnarray}

To this we add 
\begin{itemize}
\item[$\diamondsuit$] the integrability conditions of (\ref{XYZO}), namely (\ref{basic_eq1}+\ref{basic_eq6},\ref{basic_eq2},\ref{basic_eq3}), 
\item[$\diamondsuit$] the three
$(\alpha \alpha)$--Einstein field equations, namely (\ref{basic_eq8},\ref{basic_eq9},\ref{basic_eq10}),
\item[$\diamondsuit$] the four equations (\ref{e2q1},\ref{e1r2},\ref{cons1},\ref{cons2}), namely (\ref{basic_eq1},\ref{basic_eq6},\ref{basic_eq4},\ref{basic_eq5}),
\end{itemize}
all simplified with the relations obtained by acting with
the $\X,\Y,\Z$ operators on (\ref{convert_b1}--\ref{convert_b3}):
\begin{eqnarray}
\fl \X(\R_2)= \frac{1}{4\O} (4\R_3 \b_3-2\Q_3\b_3-\Q_1\b_1+\R_2\b_2-\B_6) -\Q_1(\R_2+\Q_2)-2\N\Q_3 +\frac{\E_{12}}{3}, \label{basic_eq1}\\
%
\fl \Y(\Q_1)= \frac{1}{4\O} (4\R_3 \b_3+2\Q_3\b_3-\Q_1\b_1+\R_2\b_2-\B_6) +\R_2(\R_1+\Q_1)+2\N\Q_3 -\frac{\E_{12}}{3}, \label{basic_eq6}\\
\fl 2\Y(\R_3)+\Z( \R_2)+\frac{1}{2\O} \Z(\b_1)= \N \left(\Q_1+\frac{\b_2}{2\O}\right)-\R_2(\Q_3-\R_3)-\frac{\b_1}{2\O} (\Q_3+3\R_3),\label{basic_eq2}\\
\fl 2 \X(\R_3)-\Z(\Q_1)-\frac{1}{2\O} \Z(\b_2) = \N \left(\R_2+\frac{\b_1}{2\O}\right)-\Q_1(\Q_3+\R_3)-\frac{\b_2}{2\O} (\Q_3-3\R_3),\label{basic_eq3}\\
\fl \X(\N)+\frac{1}{4\O} \X(\b_3)+\Z(\Q_2)= -\left(\N+\frac{3\b_3}{4\O}\right) \Q_1+(\Q_3-\R_3) \Q_2 +\R_2 \Q_3+\N \R_1-\R_2\R_3 \nonumber \\ -\frac{\b_2\b_3}{8\O^2}-\smfrac{1}{3}\E_{23} ,\label{basic_eq4}\\
\fl \Y(\N)+\frac{1}{4\O} \Y(\b_3)+\Z(\R_1)= \left(\N+\frac{3\b_3}{4\O}\right) \R_2-(\Q_3+\R_3) \R_1 +\R_2 \Q_3-\N \Q_2-\Q_1\Q_3 \nonumber \\ +\frac{\b_1\b_3}{8\O^2}+\smfrac{1}{3}\E_{13} ,\label{basic_eq5}\\
\fl \X(\R_1)-\Y(\Q_2)= \R_1^2+\Q_2^2+\R_3^2-\Q_3^2-\frac{\J}{9}-\frac{\E_3}{3}+\frac{\N\b_3}{2\O}+\smfrac{3}{16}\frac{\b_3}{\O^2} ,\label{basic_eq8} \\
\fl \Y(\R_2)-\Z(\Q_3+\R_3)= (\Q_3+\R_3)^2+\R_2^2-\Q_1\R_1-\frac{\J}{9}+\frac{\E_3-\E_0}{6}-\smfrac{1}{16}\frac{\b_3^2}{\O^2},\label{basic_eq9} \\
\fl \X(\Q_1)-\Z(\Q_3-\R_3)= -(\Q_3-\R_3)^2-\Q_1^2+\Q_2\R_2+\frac{\J}{9}-\frac{\E_3+\E_0}{6}+\smfrac{1}{16}\frac{\b_3^2}{\O^2},\label{basic_eq10} \\
%
\fl \Z(\b_3)=\R_2\b_2-\Q_1\b_1-\B_6,\label{basic_eq7}
\end{eqnarray}

$\diamondsuit$ \ the equations obtained by evaluation of $\X(\ref{convert_Q3R3}_a)$ and $\Y(\ref{convert_Q3R3}_b)$,
\begin{eqnarray}
\fl \X(\Q_3+\R_3) = 2\Q_3\R_1-\smfrac{3}{4}\frac{\b_3\R_2}{\O}+\smfrac{1}{4}\frac{\b_2(\Q_3+\R_3)}{\O}-\frac{\E_{13}}{3}+\frac{\B_8}{4 \O}-
\smfrac{3}{16}\frac{\b_1\b_3}{\O^2},\label{basic_eq19} \\
\fl \Y(\Q_3-\R_3) = -2\Q_3\Q_2+\smfrac{3}{4}\frac{\b_3\Q_1}{\O}+\smfrac{1}{4}\frac{\b_1(\R_3-\Q_3)}{\O}+\frac{\E_{23}}{3}+\frac{\B_7}{4 \O}+
\smfrac{3}{16}\frac{\b_2\b_3}{\O^2},\label{basic_eq20}
\end{eqnarray}

$\diamondsuit$ \ and the $\w{\nabla \cdot E}$ Bianchi equations\footnote{the $\w{\nabla}\cdot \w{H}$ equations are identities under the 
$\O$-integrability conditions, (\ref{basic_eq1}+\ref{basic_eq6},\ref{basic_eq2},\ref{basic_eq3})} (\ref{divE1},\ref{divE2},\ref{divE3}):
\begin{eqnarray}
\fl \smfrac{1}{2} \X(\E_0-\E_3+\smfrac{2}{3}\J)+\Y(\E_{12})+\Z(\E_{13})= (\R_1 -\smfrac{1}{2}\Q_1)\E_0+\smfrac{3}{2}\Q_1\E_3+\R_2 \E_{12}-
2 \Q_2\E_{12}\nonumber \\
+(\Q_3-3\R_3)\E_{13}+\left(\N +\frac{\b_3}{\O}\right)\E_{23} ,\label{basic_eq21} \\
\fl \smfrac{1}{2}\Y(\E_0+\E_3-\smfrac{2}{3}\J)-\X(\E_{12})-\Z(\E_{23})= (-\Q_2 +\smfrac{1}{2}\R_2)\E_0+\smfrac{3}{2}\R_2\E_3+\Q_1 \E_{12}-
2 \R_1\E_{12}\nonumber \\
+(\Q_3+3\R_3)\E_{23}+\left(\N +\frac{\b_3}{\O}\right)\E_{13} , \label{basic_eq22} \\
\fl \X(\E_{13})+\Y(\E_{23})+\Z(\E_3+\smfrac{1}{3}\J) = (\R_1-2\Q_1)\E_{13}+(2\R_2-\Q_2) \E_{23} \nonumber \\
-\Q_3\E_0-3\R_3\E_3 . \label{basic_eq23} 
\end{eqnarray}

\section*{Acknowledgement}

All calculations were performed using the Maple 2015 symbolic algebra package and checked with Mathematica 7.0.

\end{document}